\documentclass[runningheads]{llncs}

\newif\ifarxiv
\arxivtrue
\newcommand{\arxivTHENgd}[2]{\ifarxiv#1\else#2\fi}

\usepackage{doi}
\usepackage{thmtools,thm-restate}
\usepackage{todonotes}
\usepackage[basic]{complexity}
\usepackage{amssymb}
\usepackage{graphicx,amsmath,amssymb,pgf}
\usepackage[capitalise]{cleveref}
\usepackage{subcaption}
\usepackage{xspace} 
\usepackage{lscape}
\usepackage{multirow}
\usepackage{booktabs}

\arxivTHENgd{}{\usepackage{nopageno}}

\newcommand{\udr}{\textsf{UDR}\xspace}
\newcommand{\udc}{\textsf{UDC}\xspace}
\newcommand{\wudc}{weak \textsf{UDC}\xspace}
\newcommand{\udrs}{\textsf{UDR}s\xspace}
\newcommand{\udcs}{\textsf{UDC}s\xspace}
\newcommand{\wudcs}{weak \textsf{UDC}s\xspace}

\newcommand{\old}[1]{{}}
\newcommand*{\Parallelogramm}[1][]{
  \pgfpicture\pgfsetroundjoin
    \pgftransformxslant{.6}
    \pgfpathrectangle{\pgfpointorigin}{\pgfpoint{.60em}{.65em}}
    \pgfusepath{stroke,#1}
  \endpgfpicture}

\crefname{question}{question}{questions}
\Crefname{question}{Question}{Questions}

\crefname{lemma}{Lemma}{Lemmas}
\Crefname{lemma}{Lemma}{Lemmas}

\crefname{observation}{Observation}{Observation}
\Crefname{observation}{Observation}{Observation}

\hypersetup{
hidelinks,
colorlinks=true,
citecolor=[rgb]{1 0 0},
linkcolor=[rgb]{1 0 1},
urlcolor=[rgb]{1 0 0}
}

\title{
Unit Disk Representations of Embedded Trees, Outerplanar and Multi-Legged Graphs 
}

\titlerunning{Unit Disk Representations}

\author{Sujoy Bhore\inst{1}\orcidID{0000-0003-0104-1659} \and
	Maarten L\"offler\inst{2} \and
	Soeren Nickel\inst{3}\orcidID{0000-0001-5161-3841} 
	\and Martin N\"ollenburg\inst{3}\orcidID{0000-0003-0454-3937}}
\authorrunning{S. Bhore et al.}

\institute{Indian Institute of Science Education and Research, Bhopal, India\\
	\email{sujoy.bhore@gmail.com}
	\and Department of Computing and Information Sciences, Utrecht University\\
	\email{m.loffler@uu.nl}
	\and Algorithms and Complexity Group, TU Wien, Vienna, Austria\\
	\email{[soeren.nickel|noellenburg]@ac.tuwien.ac.at }}
\begin{document}

\maketitle

\begin{abstract}
    A unit disk intersection representation (\udr) of a graph $G$ represents each vertex of $G$ as a unit disk in the plane, such that two disks intersect if and only if their vertices are adjacent in $G$.
    A \udr with interior-disjoint disks is called a unit disk contact representation (\udc).
    We prove that it is \NP-hard to decide if an outerplanar graph or an embedded tree admits a \udr. We further provide a linear-time decidable characterization of caterpillar graphs that admit a \udr.
    Finally we show that 
    it can be decided in linear time if a lobster graph admits a \emph{weak}
    \udc, which permits intersections between disks of non-adjacent vertices.
    
\end{abstract}

\section{Introduction}\label{sec:intro}
The representation of graphs as contacts or intersections of disks has been a major topic of investigation in geometric graph theory and graph drawing. 
The famous circle packing theorem states that every planar graph has a contact representation by touching disks (of various size) and vice versa~\cite{koebe1936kontaktprobleme}. 
Since then, a large body of research has been devoted to the 
representation of planar graphs as contacts or intersections of various kinds of geometric objects~\cite{chalopin2010planar,chaplick2012planar,felsner2013rectangle,gonccalves2018planar}.
In this paper, we are interested in unit-radius disks.
A set of unit disks in $\mathbb{R}^2$ is a \emph{unit disk intersection representation} (\udr) of a graph $G=(V,E)$, if there is a bijection between~$V$ and the set of unit disks such that two disks intersect if and only if they are adjacent in $G$. 
\emph{Unit disk graphs} are graphs that admit a \udr.
\emph{Unit disk contact graphs} (also known as \emph{penny graphs}) are the subfamily of unit disk graphs that have a \udr with interior-disjoint disks, which is thus called a \emph{unit disk contact representation} (\udc). This can be relaxed to \emph{weak} \udcs, which permit contact between non-adjacent disks; see \cref{fig:intro} for examples.

\begin{figure}
	\begin{minipage}[t]{.3\linewidth}
		\centering
		\includegraphics[page=1]{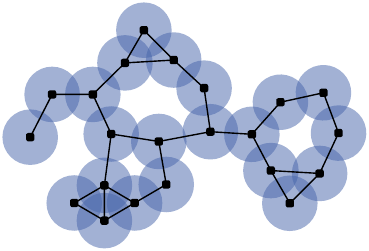}
		\subcaption{Outerplanar graph}
		\label{fig:intro_a}
	\end{minipage}
	\quad
	\begin{minipage}[t]{.3\linewidth}
		\centering
		\includegraphics[page=2]{figures/intro_figure.pdf}
		\subcaption{Caterpillar graph}
		\label{fig:intro_b}
	\end{minipage}
	\quad
	\begin{minipage}[t]{.3\linewidth}
		\centering
		\includegraphics[page=3]{figures/intro_figure.pdf}
		\subcaption{Lobster graph}
		\label{fig:intro_c}
	\end{minipage}
	\caption{
	    We investigate specific contact and intersection graphs of unit disks.
		In a \udr (a-b) disks are allowed to overlap, and contact of two disks implies an edge between their vertices.
		In a \wudc the disks are interior disjoint, but contact between non-adjacent disks is allowed.
		The disks of backbone vertices are colored grey (b-c).
	}
	\label{fig:intro}
\end{figure}

The recognition problem, where the objective is to decide  whether a given graph admits a \udr, has a rich history~\cite{bdlrst-rscplrudct-15,breu1995complexity,hlinveny1998classes,hlinveny2001representing}. 
Breu and Kirkpatrick~\cite{bk-udgrn-98} proved that it is \textsf{NP}-hard to decide whether a graph $G$ admits a \udr or a \udc, even for planar graphs.
Klemz et al.~\cite{knp-rwdcg-15a} showed that recognizing outerplanar unit disk contact graphs is already \textsf{NP}-hard, but it is decidable in linear time for caterpillars, i.e., trees whose internal vertices form a path (see \cref{fig:intro_b}). 

Recognition with a fixed embedding is an important variant of the recognition problem. Given a plane graph $G$, the objective in this problem is to decide whether $G$ admits a \udc in the plane that preserves the cyclic order of the neighbors at each vertex. Some recent works investigated the recognition problem of \textsf{UDCs} with/without fixed embedding, and narrowed down the precise boundary between hardness and tractability; see~\cite{bdlrst-rscplrudct-15,ccn-recwwudcrn-19,cleve-20,knp-rwdcg-15a}. A remaining open question is to settle the complexity of recognizing non-embedded trees that admit a \udc. Some of these works focused on \emph{weak} \textsf{UDCs}, where disks of non-adjacent vertices may touch. In  this model, the recognition of non-embedded trees that admit a weak \udc is \NP-hard~\cite{cleve-20}. We summarize the results on \wudcs in  Table~\ref{tab:results_state_open}.  

While several results of the past years have shed  light into the recognition complexity gap for \udcs, not much is known in this regard for the more general class of \udrs since the \NP-hardness  for planar graphs from 1998~\cite{bk-udgrn-98}. 
\paragraph{Our Contribution.}
We investigate the unit disk graph recognition problem for subclasses of planar graphs.
We show that recognizing unit disk graphs remains \textsf{NP}-hard for non-embedded outerplanar graphs (see \cref{fig:intro_a}) -- strengthening the previous hardness result for planar graphs~\cite{bk-udgrn-98} -- and for embedded trees (Section~\ref{hardness-section}).
This line of research aims to extend earlier investigations~\cite{bdlrst-rscplrudct-15,ccn-recwwudcrn-19,cleve-20} of \textsf{UDC}s to the \textsf{UDR} model and builds in particular on the work of Bowen et al.~\cite{bdlrst-rscplrudct-15}.

On the positive side, we provide a linear-time algorithm to recognize caterpillar graphs (see \cref{fig:intro_b}) that admit a \udr (Section~\ref{algo-section:caterpillar}). 
In Section~\ref{algo-section:lobster}, we return to the problem of recognizing unit disk contact graphs and extend the tractability boundary for non-embedded graphs. While it was known that a weak \udc for caterpillar graphs can be constructed in linear time (if one exists), but the same recognition problem is \NP-hard for trees~\cite{cleve-20}, we prove that we can decide in linear time 
if a lobster graph admits a weak \udc on the triangular grid, where a \emph{lobster} is a tree whose internal vertices form a caterpillar (see \cref{fig:intro_c}).
\cref{tab:results_state_open} summarizes our results and remaining open problems.
Proofs and details of results marked with a star ($\star$) can be found \arxivTHENgd{in the appendix}{in the complete version~\cite{bhore2021recognition}}.

\begin{table}[tbp]
	\begin{center}
		\begin{tabular}{lcccc}
			\toprule
			
		    graph class & \multicolumn{2}{c}{\wudc} & \multicolumn{2}{c}{\udr}\\
			 & non-embedded & embedded & non-embedded & embedded \\
			\midrule
			planar & $\uparrow$ \NP-hard & $\uparrow$ \NP-hard  &\NP-hard \cite{bk-udgrn-98}  &  $\uparrow$ \NP-hard\\
			 outerplanar & $\uparrow$ \NP-hard & $\uparrow$ \NP-hard &\NP-hard (Thm.\ref{thm:outer_hardness})  &  $\uparrow$ \NP-hard  \\
			 trees & \NP-hard \cite{cleve-20} & $\uparrow$ \NP-hard  & open & \NP-hard (Thm.\ref{thm:tree_hardness})\\
			 lobsters & $O(n)$ (Thm.\ref{thm:dyn_prog_lobsters}) & $\uparrow$ \NP-hard & open & open \\
			 caterpillars & $O(n)$ \cite{cleve-20} & \NP-hard \cite{ccn-recwwudcrn-19} &$O(n)$ (Thm.\ref{algo-main-th}) & open \\
			\bottomrule
		\end{tabular}
		
	\end{center}
	\caption{State of the art, our results, and open problems on unit disk graph recognition. Upward arrows indicate, that a result follows from the one below.}
	\label{tab:results_state_open}
\end{table}

\section{Preliminaries}\label{app:prelim}
A graph $G = (V, E)$ with $V = \{v_1, \dots, v_n\}$ is a unit disk graph if there exists a set of closed unit disks $\mathcal{D} = \{d_1, \dots, d_n\}$ and a bijective mapping $d \colon V \rightarrow \mathcal{D}$, s.t., $d(v_i) = d_i$ and $v_iv_j \in E$ if and only if $d_i$ and $d_j$ intersect.
We call $\mathcal D$ a \emph{unit disk intersection representation} (\udr) of $G$.
If all disks in $\mathcal{D}$ are pairwise interior disjoint we also call $\mathcal D$ a \emph{unit disk contact representation} (\udc) of $G$.
A graph is a \emph{unit disk contact graph} if it admits a \udc.
A \emph{weak} \udc permits contact between two disks $d_i$ and $d_j$, even if $v_iv_j \not \in E$.\footnote{Note that weak \udrs, in contrast, are not interesting, since a complete graph $K_n$ can easily be represented as a \udr and therefore every graph admits a weak \udr.}
A \emph{caterpillar} graph $G$ is a tree, which yields a path, when all leaves are removed.
We call this path, the backbone $B_G$ of $G$.
Similarly a \emph{lobster} graph $G'$ is a tree, which yields a caterpillar graph $G''$, when all leaves are removed.
The backbone of $G'$ is the backbone of $G''$.
For each vertex $v$ of the backbone,  we denote the set of vertices that are reachable from $v$ on a path that does not include any other backbone vertex, as the \emph{descendants} of $v$. 

The set of disks $\mathcal{D}$ induces an embedding $\mathcal{E}_\mathcal{D}(G)$ of $G$ by placing every vertex $v$ at the center of $d(v)$ and linking neighboring vertices by straight-line edges.
We will therefore also use $v$ as the center of $d(v)$.
We say that a \udr $\mathcal{D}$ preserves the embedding of an embedded graph $G$ with embedding $\mathcal{E}(G)$ if $\mathcal{E}_\mathcal{D}(G) = \mathcal{E}(G)$.
Let $a,b,c$ be three points in $\mathbb R^2$.
We use $\measuredangle abc$ to denote the clockwise angle defined between the segments ${ab}$ and ${bc}$. 
We define the angle $\measuredangle d_id_jd_k$ as the clockwise angle $\measuredangle v_iv_jv_k$ in $\mathcal{E}_\mathcal{D}(G)$.

\section{\NP-Hardness Results}\label{hardness-section}

\begin{figure}[tb]
	\centering
	\begin{minipage}[t]{.4\linewidth}
		\centering
		\includegraphics[width=\linewidth]{figures/BowenAuxilliary.pdf}
		\subcaption{Auxiliary hexagonal grid structure with embedded incidence graph}
		\label{fig:bowen_auxilliary_a}
	\end{minipage}
	\hfill
	\begin{minipage}[t]{.55\linewidth}
		\centering
		\includegraphics[page=1, width=\linewidth]{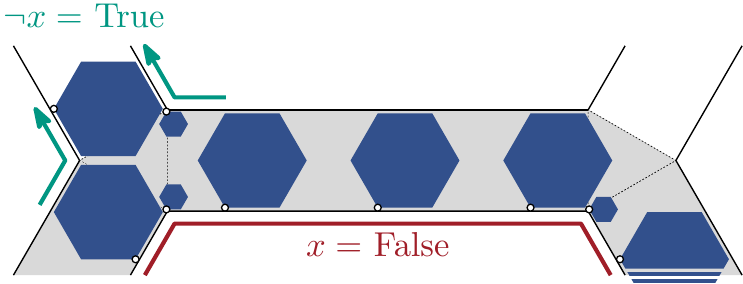}
		\subcaption{False variable $x$. The connected wire on the top left is in a true state and therefore corresponds to the negated literal $\neg x$.}
		\label{fig:bowen_auxilliary_b}
	\end{minipage}
	\begin{minipage}[t]{.4\linewidth}
		\centering
		\includegraphics[page=3, width=\linewidth]{figures/BowenFill.pdf}
		\subcaption{Unsatisfied state (left) and one possible satisfied state (right) of the clause gadget}
		\label{fig:bowen_auxilliary_c}
	\end{minipage}
	\hfill
	\begin{minipage}[t]{.55\linewidth}
		\centering
		\includegraphics[page=2, width=\linewidth]{figures/BowenFill.pdf}
		\subcaption{True variable $x$. The wire is in a false state. The dotted hexagon indicates the overlap, if hexagons would have inconsistent states.}
		\label{fig:bowen_auxilliary_d}
	\end{minipage}
	\caption{
		Auxiliary structure details used by Bowen et al.~\cite{bdlrst-rscplrudct-15}. All Figures are recreations/adaptions from their paper.
		The incidence graph is embedded on a hexagonal grid (a).
		The edges are short corridors in which the blue hexagons are fitted, hinged at white vertices.
		Hexagons in the variable cycle (red line, grey backdrop) have two states (b) and (d).
		The clause gadget (c) requires one hexagon, which does not enter the junction.
	}
	\label{fig:bowen_auxilliary}
\end{figure}

In this section, we prove that recognizing unit disk graphs remains \NP-hard for non-embedded outerplanar graphs and for embedded trees. Our proofs apply the generic machinery of Bowen et al.~\cite{bdlrst-rscplrudct-15} to decide realizability of polygonal linkages, which requires to construct gadgets that can model hexagons and rhombi in a stable way. First we sketch the reduction of Bowen et al.\ (full details in
\arxivTHENgd{Appendices~\ref{sec:poly_linkage-appendix}--\ref{sec:approx_construction-appendix}}{Appendices~A.1-A.3 of the full version~\cite{bhore2021recognition}}),
then we describe our constructions of the required stable structures with outerplanar (Section~\ref{sec:outer}) and embedded tree gadgets (Section~\ref{sec:trees}).

Bowen et al.~\cite{bdlrst-rscplrudct-15} proved that recognizing unit disk contact graphs is \textsf{NP}-hard for embedded trees, via a reduction from planar 3-SAT, which uses an auxiliary construction formulated as a realization of a polygonal linkage.
A polygonal linkage is a set of polygons, which are realizable if they can be placed in the plane, s.t., predefined sets of points on the boundary of these polygons are identified.
Bowen et al.\ define a set of hexagons in a hexagonal tiling (\cref{fig:bowen_auxilliary_a}) with small gaps between them, which form a hexagonal grid, in which a representation of the incidence graph of the planar 3-SAT instance is fitted.
Smaller hexagons are fitted into cycles in this grid, s.t., they admit only two different realizations, see \cref{fig:bowen_auxilliary_b}, and determine the state of neighboring small hexagons, see \cref{fig:bowen_auxilliary_d}. The cycles represent variables in a true or false state.
The states of the cycles are transmitted via chains of smaller hexagons into the gaps.
The vertex, where three such chains meet, contains a small hexagon on a thin connection, which can only be realized if at least one transmitted state is \textit{true}, see \cref{fig:bowen_auxilliary_c}. The polygonal linkage is realizable if and only if the planar 3-SAT instance is satisfiable. For a detailed description, we refer to \arxivTHENgd{\cref{hardness-appendix}}{Appendix A of the full version~\cite{bhore2021recognition}} and Bowen et al.~\cite{bdlrst-rscplrudct-15}.
The building blocks of this reduction are hexagons of variable sizes and short segments.
Bowen et al.~\cite{bdlrst-rscplrudct-15} approximate the hexagons by creating graphs, whose \textsf{UDC}s must be within a constant (asymmetric) Hausdorff distance\footnote{Recall that the asymmetric Hausdorff distance from a set $A$ to a set $B$ in a metric space with metric $d$ is defined as $d_H(A,B) = \sup_{a \in A} \inf_{b \in B} d(a,b)$.} 
of a hexagon and the segments similarly by providing graphs, whose \textsf{UDC}s must be within a constant (asymmetric) Hausdorff distance of long thin rhombi. 
We extend their notion of $\lambda$-stable approximations to \textsf{UDR}s.
\begin{definition}
	A graph $G$ is a $\lambda$-stable approximation of a region $P$ in the plane if, in every \textsf{UDR} of $G$, there exists a congruence transformation $f:\mathbb{R}^2\rightarrow \mathbb{R}^2$ such that the union $U$ of all unit disks in the \textsf{UDR} has an asymmetric Hausdorff distance $d_H(f(P),U) \le \lambda$.
\end{definition}

\subsection{Non-embedded Outerplanar Graphs}
\label{sec:outer}

We prove that recognition of unit disk graphs is hard for non-embedded outerplanar graphs by providing outerplanar graphs $G^H_k$ and $G^R_k$, which are $7$-stable approximations of a hexagon of side length $2k-1$ and a rhombus of width $2k+6$, respectively.
Then, the \NP-hardness follows immediately from the construction of Bowen et al.~\cite{bdlrst-rscplrudct-15} sketched above.
To obtain the $7$-stable approximations we present two graphs, which enforce local bends in one direction of at least $\pi$ and $\frac{4\pi}{3}$ in any \textsf{UDR}.

\begin{figure}
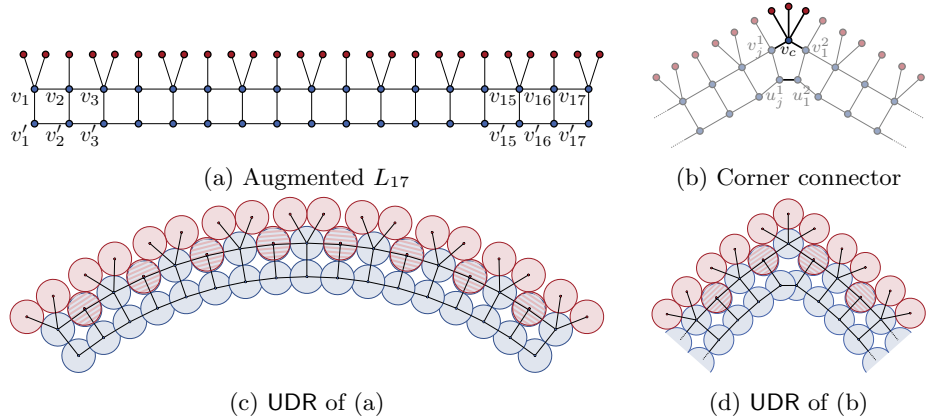

	\centering	
	\begin{minipage}[t]{.65\linewidth}
		\centering
		\includegraphics[page=2,width=\linewidth]{figures/red_outer_hexagon.pdf}
		\subcaption{Augmented $L_{17}$}
		\label{fig:red_outer_graph_a}
	\end{minipage}
	\hfill
	\begin{minipage}[t]{.3\linewidth}
		\centering
		\includegraphics[page=1,width=\linewidth]{figures/red_outer_hexagon.pdf}
		\subcaption{Corner connector}
		\label{fig:red_outer_graph_b}
	\end{minipage}
	
	\begin{minipage}[t]{.65\linewidth}
		\centering
		\includegraphics[page=4,width=\linewidth]{figures/red_outer_hexagon.pdf}
		\subcaption{\udr of (a)}
		\label{fig:red_outer_graph_c}
	\end{minipage}
	\hfill
	\begin{minipage}[t]{.3\linewidth}
		\centering
		\includegraphics[page=3,width=\linewidth]{figures/red_outer_hexagon.pdf}
		\subcaption{\udr of (b)}
		\label{fig:red_outer_graph_d}
	\end{minipage}
	\caption{
		Components for creating $\lambda$-stable approximations $G^H_k$ and $G^R_k$.
		Ladder $L_{17}$ (a) and its \udr (c), as well as a corner connector (b)  and its \udr (d).
		The corner connector connects two ladders (lighter colored parts) and its actual parts have a darker color in (b). 
		The bends in the \udrs are required but exaggerated.
	}
	\label{fig:red_outer_graph}
\end{figure}

A \emph{ladder} $L_k$ (see \cref{fig:red_outer_graph_a,fig:red_outer_graph_c}) is a chain of pairwise connected vertices $v_i$ and $v_i'$, for $i=1, \dots, k$, also called the \textit{outer} and \textit{inner} vertices of $L_k$, respectively.
Additionally, so-called \textit{extension neighbors}, which are connected to only one outer vertex each, are added on one side of the ladder, s.t., the outer vertices have alternating degrees of four and five.
Since the ladder consists of a chain of $C_4$'s, these neighbors are forced to be placed all on the outside in an outerplanar embedding.
The minimal height of such a ladder is $2\sqrt{3}+2-\varepsilon$, which is the height of the smallest bounding box of a tight packing of three rows of unit disks minus an (arbitrarily) small constant $\varepsilon>0$.

The permission of overlap between adjacent disks allows for a placement of one extension neighbor almost on top of its adjacent outer vertex, which leads to an ever so slight inwards bend and, more importantly, any outwards bend is impossible.
In order to enforce an inward bend of at least $\frac{4\pi}{3}$, we connect the last inner vertex $u_j^1$ of one ladder with the first inner vertex $u_1^2$ of a second ladder and the last and first outer vertex $v_j^1$/$v_1^2$ of the first and second ladder, respectively, both with a vertex $v_c$, which has three attached extension neighbors, see \cref{fig:red_outer_graph_b}.
This construction is called a \emph{corner connector}.
Since it is impossible to place the disk of any extension neighbor of $v_c$ inside the 5-cycle $d(v_j^1), d(u_j^1), d(u_1^2), d(v_1^2), d(v_c)$, without overlapping at least two disks in the \textsf{UDR}, all extension neighbors are still forced to the outside and therefore $\measuredangle v_j^1v_cv_1^2 > \frac{4\pi}{3}$, see \cref{fig:red_outer_graph_d}.

Placing two ladders opposite each other and connecting them on one end with three corner connectors as shown in \cref{fig:red_outer_rhombus}, yields a $7$-stable approximation $G^R_k$ of a thin rhombus. 
The following lemma is analogue to Lemma 10 in~\cite{bdlrst-rscplrudct-15}.

\begin{restatable}[$\star$]{lemma}{rhombusouter}\label{lem:rhombus_approx}
	For every positive integer $k$ the outerplanar graph $G^R_k$ in \cref{fig:red_outer_rhombus} is a 7-stable approximation of a rhombus of width $2k+6$ and height $6\sqrt{3}+2$.
\end{restatable}
\begin{proof}[Sketch]
	The graph $G^R_k$ is made up of components, which make any bend to the outside impossible, see \cref{fig:red_outer_rhombus_a}.
	Since two ladders need to be overlap free, the minimum height of the ladders guarantees that part of the boundary of the union over all disks in the \textsf{UDR} of $G^R_k$ lies above and below the line $\overline{c_0l_0}$, see \cref{fig:red_outer_rhombus_b}.
	The largest vertical distance of the rhombus to this line is $3\sqrt{3} + 1 \approx 6.196 < 7$.
\end{proof}

\begin{figure}
	\centering
	\begin{minipage}[t]{0.5\linewidth}
		\centering
		\includegraphics[page=3, width=\linewidth]{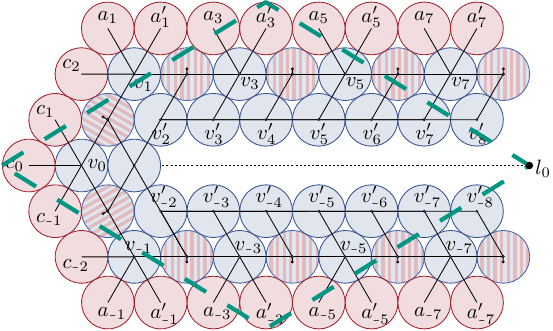}
		\subcaption{Maximal distance to $\overline{c_0l_0}$}
		\label{fig:red_outer_rhombus_a}
	\end{minipage}
	\hfill
	\begin{minipage}[t]{0.42\linewidth}
		\centering
		\includegraphics[page=4, width=\linewidth]{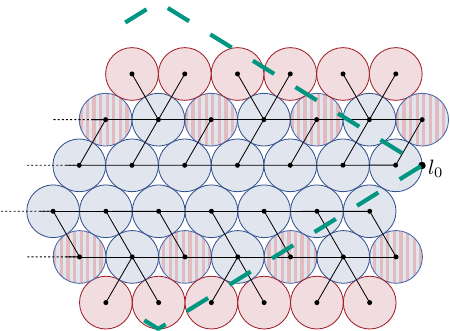}
		\subcaption{\udr on both sides of $\overline{c_0l_0}$}
		\label{fig:red_outer_rhombus_b}
	\end{minipage}
	\caption{A $7$-stable approximation $G^R_7$ of a rhombus superimposed on its \udr. The maximal distance of any point of the \udr to $\overline{c_0l_0}$ is smaller or equal 7 (a) and at all points of $\overline{c_0l_0}$ a part of the \udr lies above and below $\overline{c_0l_0}$ (b).
		Both \udrs require a small inward bend to be valid and hatched disks indicate almost overlapping placement of a red disk on a blue disk, with a small shift to the outer side.
		Both inward bends and outward shifts are omitted.
	}
	\label{fig:red_outer_rhombus}
\end{figure}

\begin{restatable}[$\star$]{lemma}{hexagonapprox}
\label{lem:hexagon_approx}
	For every integer $k = 6n+4, n\in \mathbb{N}$, the outerplanar graph $G^H_k$ in \cref{fig:red_complete_hex} is a $7$-stable approximation of a regular hexagon of side length $2k-1$.
\end{restatable}
\begin{proof}[Sketch]
    The $7$-stable approximation $G^H_k$ of a hexagon of side length $2k-1$ uses again ladders and corner connectors to trace the outline of the hexagon.
    Then the inside of this construction is filled with a set of ladders 
    until no additional ladders can be added, see \cref{fig:red_complete_hex}.
    Outwards bends are impossible by construction of the ladders and corner connectors and inwards bends are strongly limited since the interior of the hexagon is almost completely filled with ladders.
    Since the amount of compression in a ladder is very limited, this leaves only a constant amount of space on the inside of a \textsf{UDR} of $G^H_k$.
\end{proof}

\begin{theorem}\label{thm:outer_hardness}
	Recognizing unit disk graphs is \NP-hard for outerplanar graphs.
\end{theorem}
\begin{proof}
This result follows from \Cref{lem:rhombus_approx,lem:hexagon_approx} by using the polygonal linkage reduction of Bowen et al.~\cite{bdlrst-rscplrudct-15} (see also \arxivTHENgd{\cref{hardness-appendix}}{Appendix A of the full version~\cite{bhore2021recognition}}).
Note that we can emulate hinges exactly as in the original reduction.
\end{proof}

\begin{figure}
	\centering
	\begin{subfigure}[t]{\linewidth}
		\centering
		\includegraphics[page=5, width=.75\linewidth]{figures/red_outer_hexagon.pdf}
		\subcaption{
			Combination of the ladder and corner connectors into $7$-stable approximation $G^H_k$ of a hexagon superimposed on its \udr.
		    \\
		}
		\label{fig:red_complete_hex}
	\end{subfigure}
	\begin{subfigure}[t]{\linewidth}
		\centering
		\includegraphics[page=8, width=.75\linewidth]{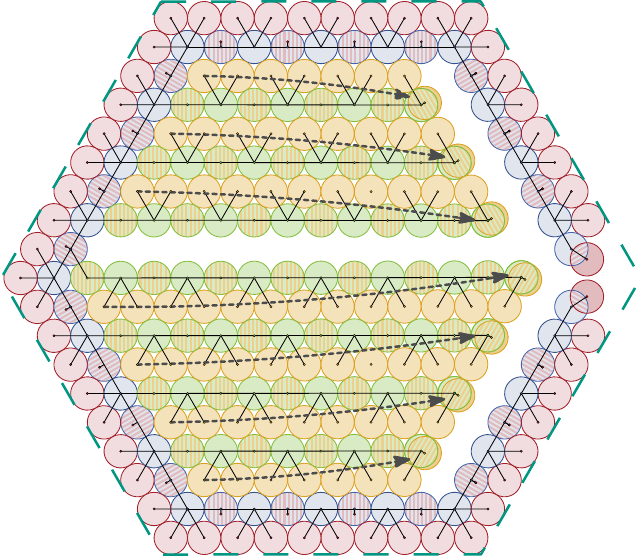}
		\caption{
			A $7$-stable approximation tree $T^H_k$ of a hexagon superimposed on its \udr.
		}
		\label{fig:red_complete_hex_tree}
	\end{subfigure}
	\caption{7-stable hexagon approximations. Hatched disks indicate almost overlapping placement of disks. Necessary infinitesimal bends are omitted. The bend directions of the inner components are indicated with gray arrows. The approximated regular hexagon is indicated by the dashed green outlines.}\label{fig:stablehex}
\end{figure}

\subsection{Embedded Trees}
\label{sec:trees}
By slightly adapting the construction of the outerplanar graphs of Section~\ref{sec:outer}, we can prove that recognizing unit disk graphs is \NP-hard for embedded trees.
The crucial observation is that we used the outerplanarity of $G^R_k$ and $G^H_k$ exclusively to be able to build a tree-like structure out of chains of $4$- and $5$-cycles.
We used this to force the placement of leaf disks to a specific side of these chains in any \udr of $G^R_k$ and $G^H_k$.
As we are concerned with embedded trees in this section, we can omit the inner vertices of the ladder, as the given embedding puts the leaves on the desired side of the chains.
This results in a tree.
We call a ladder without the inner vertices a \emph{chain}.

We can now use a very similar construction idea as for $G^R_k$ and $G^H_k$ above. We need to augment both gadgets with an additional chain 
in order to retain the property, that no parts of these gadgets can be folded onto themselves.
The resulting trees $T^R_k$ and $T^H_k$ are shown in \cref{fig:red_tree_rhombus,fig:red_complete_hex_tree}.
\arxivTHENgd{A more detailed description can be found in the proof of \cref{thm:tree_hardness} in \cref{sec:app:tree_hardness}.}{}
\begin{restatable}[$\star$]{theorem}{treehardness}
    \label{thm:tree_hardness}
	Recognizing unit disk graphs is \NP-hard for embedded trees. 
\end{restatable}

\begin{figure}[tb]
	\centering
	\begin{minipage}[t]{0.44\linewidth}
		\centering
		\includegraphics[page=3, scale=.6]{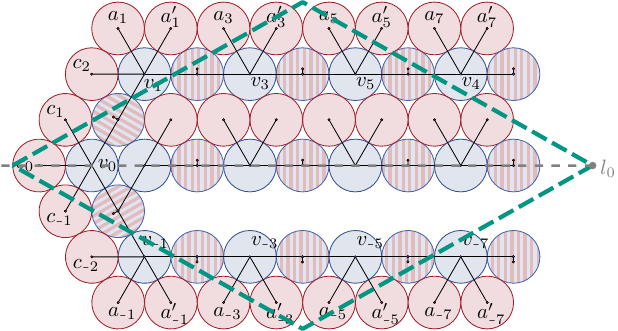}
		\subcaption{Maximal distance to $\overline{c_0l_0}$}
		\label{fig:red_tree_rhombus_a}
	\end{minipage}
	\hfill
	\begin{minipage}[t]{0.44\linewidth}
		\centering
		\includegraphics[page=5,scale=.6]{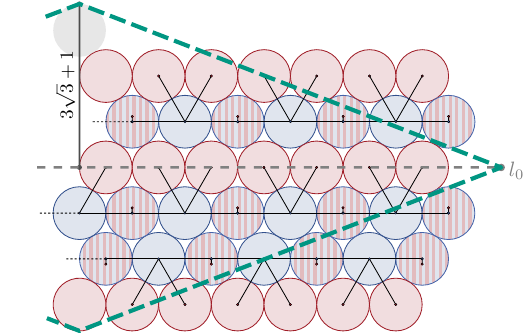}
		\subcaption{\udr on both sides of $\overline{c_0l_0}$}
		\label{fig:red_tree_hexagon}
	\end{minipage}
	\caption{
		A $7$-stable approximation $T^R_k$ of a long thin rhombus superimposed on its \udr. The maximal distance of any point of the \udr to $\overline{c_0l_0}$ is smaller or equal 7 (a) and at all points of $\overline{c_0l_0}$ a part of the \udr lies above and below $\overline{c_0l_0}$ (b).
		In both cases at any point along $\overline{c_0l_0}$ at least one point on the boundary of the union of all disks in a \udr of $T^R_k$ lies on or above $\overline{c_0l_0}$ and on or below $\overline{c_0l_0}$.
	}
	\label{fig:red_tree_rhombus}
\end{figure}

\section{Recognition Algorithm for Caterpillars}\label{algo-section:caterpillar}
We propose a linear-time algorithm using similar ideas to Klemz et al.~\cite{knp-rwdcg-15a}, that recognizes if an input caterpillar graph $G=(V,E)$
admits a \udr or not; it is constructive and provides a representation if one exists.  
However, we need to address several new issues as we show that a larger class of graphs admits a \udr compared to a \udc. 
Clearly, if $G$ contains a vertex of degree at least $6$, then due to the unit disk packing property, it does not admit a \udr.
Hence, every realizable  caterpillar must have maximum
degree $\Delta \le 5$. Moreover, it is easy to observe that all caterpillars with $\Delta \le 4$ admit a \udc (and thus a \udr), as also noted by Klemz et al.~\cite{knp-rwdcg-15a}. 
Not every caterpillar with $\Delta=5$, however, is realizable as \udr. 
We show that two consecutive degree-$5$ vertices on $B_G$ cannot be realized. 
The following lemma gives a sufficient condition for a ``No'' instance to be used in the recognition algorithm. 
\arxivTHENgd{The proof of \cref{lem:claim_algo} is stated in \cref{sec:proof_characterization}.}{}

\begin{restatable}[$\star$]{lemma}{claimalgo}
	\label{lem:claim_algo}
	If $B_G$ contains two adjacent degree~$5$ vertices $u,v$, 
	then it does not admit a unit disk intersection representation.
\end{restatable}

\subsection{The Algorithm}

As a preprocessing step we augment all backbone vertices of degree~$3$ or lower with additional degree-$1$ neighbors, s.t., they have degree~$4$.
Consider a chain $v_1, \dots, v_n$ of backbone vertices.
Now assume all vertices are of degree~$4$ or lower.
We place them on a horizontal line.
For each $1 \leq i \leq n$ at disk $d(v_i)$,
we place its leaf neighbor disks $d(v_i^t), d(v_i^b)$ first at the top and then at the bottom of $d(v_i)$, see \cref{fig:caterpillar_construction_a}, s.t., the clockwise angle $\measuredangle v_i^tv_iv_i^b = \frac{4\pi}{3} - 2i\varepsilon$.
The rotational $\varepsilon$ offset avoids adjacencies between the leaf disks. While these offsets can add up, we can choose $\varepsilon$ small enough for every finite caterpillar, s.t., this is negligible.

Now we assume that not all vertices are of degree 4 or lower.
To keep the entire construction of the backbone $x$-monotone, whenever we encounter a degree~$5$ vertex $u$ after a degree four vertex $v_k$, we place $d( u^{t'})$ of its additional leaf $u^{t'}$ alternatingly on the top or the bottom side with a $\frac{\pi}{3} + \varepsilon$ rotational offset to $d(u^t)$ (or $d(u^b)$).
We will assume that we placed the disk at the top.
Therefore $\measuredangle u^{t'}uu^{b} \leq \pi - (2k+1)\varepsilon$, i.e., an almost horizontal connection, see \cref{fig:caterpillar_construction_b}.

If the next vertex $x$ has also degree five, 
then due to Lemma~\ref{lem:claim_algo} we know that the sequence is not realizable. Otherwise, we place $d(x)$, s.t., it is touching $d(u)$ with a $\frac{\pi}{3} + \varepsilon$ rotational offset to $d(x)$, 
see \cref{fig:caterpillar_construction_b}.
We place $d(x^b)$ at the planned position relative to $d(x)$ at the bottom, i.e., with a $\frac{\pi}{3} + (k+2)\varepsilon$ counterclockwise offset relative to the $x$-axis, however, we place $d(x^t)$ almost exactly on top of $d(x)$ with a very small shift of $\frac{\varepsilon}{Cn}$ orthogonal to $\overline{ux}$, 
for some large constant $C$.
This prevents touching of $d(u)$ and $d(x^t)$, without creating an adjacency between $d(u^{t'})$ and $d(x^t)$.

\begin{figure}[tb]
	\centering
	
	\begin{subfigure}[t]{0.49\linewidth}
		\centering
		\includegraphics[page=6, scale=.8]{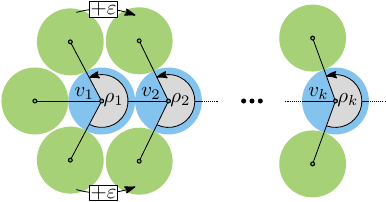}
		\caption{$\varepsilon$-offsets. $\rho_k = \frac{4\pi}{3} - 2k\varepsilon$}
		\label{fig:caterpillar_construction_a}
	\end{subfigure}
	\hfill
	\begin{subfigure}[t]{0.49\linewidth}
		\centering
		\includegraphics[page=2, scale=.8]{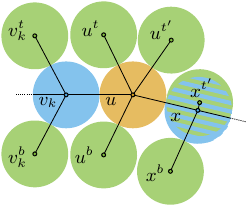}
		\subcaption{Placement of $x$ after encountering $u$}
		\label{fig:caterpillar_construction_b}
	\end{subfigure}
	
	\begin{subfigure}[t]{0.49\linewidth}
		\centering
		\includegraphics[page=3, scale=.8]{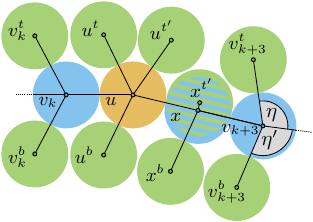}
		\caption{$\eta = \eta'$ and $ \eta + \eta' = \frac{4\pi}{3}-(2k+6)\varepsilon$}
		\label{fig:caterpillar_construction_c}
	\end{subfigure}
	\hfill
	\begin{subfigure}[t]{0.49\linewidth}
		\centering
		\includegraphics[page=4, scale=.8]{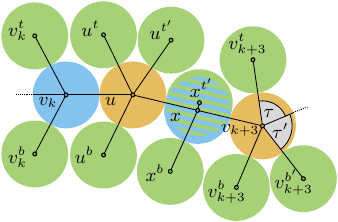}
		\subcaption{$\tau = \tau'$ and $\tau + \tau' = \pi-(2k+7)\varepsilon$}
		\label{fig:caterpillar_construction_d}
	\end{subfigure}
	\caption{
		Chains of degree-4 vertices are placed in a dense packing formation with small offsets (a).
		A degree-5 vertex places an additional leaf on one side (b).
		The next vertex $v_{k+3}$ can again be placed with the desired angle of just over $\frac{2\pi}{3}$ between two neighbors (c).
		Placement of $v_{k+3}$ is possible if its degree is 5 (d).
        Note that, the rotational offset angles are exaggerated, for better readability.
	}
	\label{fig:caterpillar_construction}
\end{figure}

From this point onwards, we consider the direction of $\overline{ux}$ to be the direction in which we extend the backbone of the caterpillar. Any following disk $d(v_{k+3})$ can be placed again in the new extension direction touching $d(x)$. Its leaf disks $d(v_{k+3}^t)$ and $d(v_{k+3}^b)$ can be placed in their planned positions,  i.e. with a clockwise or counterclockwise offset of $\frac{\pi}{3} + (k+3)\varepsilon$ relative to the new extension direction, respectively, which results in a clockwise angle $\measuredangle u_tuu_b \leq \frac{4\pi}{3} - (2k+6)\varepsilon$.
Note that $v_{k+3}$ can have a degree of four  (\cref{fig:caterpillar_construction_c}) or five (\cref{fig:caterpillar_construction_d}) and that at this point, if $v_{k+3}$ has degree five, we can immediately repeat this procedure.

As a  postprocessing step, we remove all degree-$1$ vertices that were added in the preprocessing step.
Then from the above description of the algorithm and the correctness analysis in \arxivTHENgd{Appendix~\ref{sec:app:algo_correct}}{Appendix~B.2 of the full version~\cite{bhore2021recognition}} we obtain the following theorem.

\begin{theorem}
	\label{algo-main-th}
	Let $G=(V,E)$ be a caterpillar graph.
	$G$ admits a \udr if and only if $G$ does not contain any two adjacent degree-$5$ vertices in the backbone path $B_G$ of $G$.
	This property can be tested in linear time and if a \udr exists then it can be constructed in linear time.
\end{theorem}

\section{Weak \udcs of Lobsters on the Triangular Grid}\label{algo-section:lobster}
We have shown that recognition of \udrs is \NP-hard for outerplanar graphs and linear-time solvable for caterpillars, which mirrors the results for \udcs and \wudcs; it leaves the recognition complexity for (non-embedded) trees as an open question for both \udrs and \udcs.
For \wudcs, however, recognition has been proven \NP-hard for trees~\cite{cleve-20}.
In order to investigate the complexity of \wudcs further, we zoom in on the gap between trees and caterpillars and investigate the graph class of lobsters.

The \emph{spine} 
of a \wudc of a lobster $G$ is the polyline defined by connecting the centers of all disks belonging to the vertices of $B_G$ in order.
A \wudc is \emph{straight}, if its spine is a straight line segment.
Similarly, a \wudc is $x$- or $y$-monotone, if its spine is $x$- or $y$-monotone.
Since we consider weak \udcs with contacts between non-adjacent disks permitted, we focus our attention on \wudcs placed on a triangular grid (similarly to previous work on \wudcs \cite{cleve-20}).

\subsection{Straight Backbone Lobsters}
Since any caterpillar $G$, admits a \wudc if and only if it admits a straight \wudc \cite{cleve-20} we investigate 
lobster graphs, which admit a straight \wudc.
These are not all lobsters, since any simple lobster graph containing a non-backbone vertex of degree 6 only admits a non-straight \wudc.
We observe that already for this restricted subclass, a greedy placement scheme similar to Cleve's approach~\cite{cleve-20} for caterpillars is not possible; again shown by an example.

\begin{figure}[tb]
	\centering
	\begin{subfigure}{.43\linewidth}
		\centering
		\includegraphics[page=1]{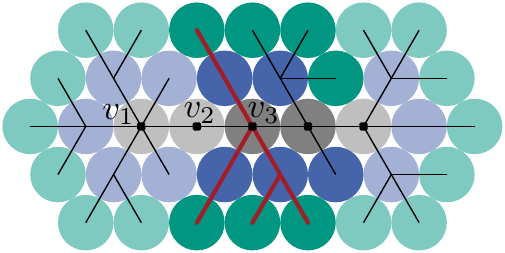}
		\subcaption{$G$ has 18 possible \wudcs}
	\end{subfigure}
	\quad
	\begin{subfigure}{.47\linewidth}
		\centering
		\includegraphics[page=2]{figures/Lobsters/not_greedy.pdf}
		\subcaption{$G'$ has 12 possible \wudcs}
	\end{subfigure}
	\caption{The subgraphs of $G$ and $G'$ induced by their first three backbone vertices are equal, however, depending on the following vertices a different realization of the neighbors of $v_3$ is necessary.}
	\label{fig:lob_not_greedy_graphs}
\end{figure}

We specify two lobster graphs $G$ and $G'$,
see \cref{fig:lob_not_greedy_graphs}.
It can be checked via exhaustive enumeration that $G$ admits 18 different \wudcs, while $G'$ admits only 12.
The subgraphs induced by their first three backbone vertices are identical, however the realization of the descendants of $v_3$ (highlighted in red) is unique for both graphs (up to symmetry) and dependent on the structure of the graphs beyond this point.
We can therefore not simply scan over the backbone in a greedy manner and fix all positions for the disks of descendants of a backbone vertex and then continue on to the next.
It is, however, still possible to do this in linear time with dynamic programming.
The requirements for this are actually less strict, as it is already sufficient to have a strictly $x$-monotone rather than a straight backbone, which we show in the next section.

\subsection{Monotone Weak UDCs}

If we can guarantee that a lobster can be realized as a strictly $x$-monotone \wudc, we can compute such a \wudc with a linear-time dynamic programming algorithm.
The dynamic program uses the following three observations.

\begin{restatable}[$\star$]{observation}{dynprogB}\label{obs:dyn_prog_2}
	The number of possible placements of a backbone vertex $v_i$ and its descendants is constant for a fixed position of $v_i$.
\end{restatable}

\begin{restatable}[$\star$]{observation}{dynprogD}\label{obs:dyn_prog_3}
    For a fixed grid position, the number of backbone vertices of a strictly $x$--monotone \wudc, who can occupy this position by themselves or with a descendant is constant. Moreover, the distance in graph between the first and the last such vertex is constant.
\end{restatable}

\begin{restatable}[$\star$]{observation}{dynprogE}\label{obs:dyn_prog_5}
    Let $C$ be a sufficiently large constant, and let $A, B$ be two \wudc, whose last $C$ backbone vertices have placed themselves and their descendants in such a manner, that the pattern of occupied grid positions and the placement of the last backbone vertex is equivalent up to translation and rotation.
    Then any extension graph $C$, which can be appended to $A$, s.t., the combined graph admits a \wudc can also be appended to $B$, s.t., their combined graph admits a \wudc and vice versa.
\end{restatable}

With these three claims we obtain the following Lemma.

\begin{restatable}[$\star$]{lemma}{dynprogmonotone}\label{lem:dyn_prog_monotone}
    Using dynamic programming it can be checked in linear time if a lobster graph admits an $x$-monotone \wudc on the triangular grid.
\end{restatable}

\subsection{General Lobsters}
The algorithm sketched in the previous section recognizes lobster graphs, which admit a strictly $x$-monotone \wudc in linear time.
Now we set out to prove that every lobster which admits a \wudc also admits a strictly $x$-monotone \wudc.
We prove this by induction.
The induction step is done as a computer-assisted proof. See \arxivTHENgd{Appendix~\ref{sec:app:automated} for details}{Appendix~D in the full version~\cite{bhore2021recognition}}.

\begin{lemma}\label{lem:x_mon_lobster}
    Every lobster graph, which admits a \wudc on the triangular grid, also admits an $x$-monotone \wudc on the triangular grid.
\end{lemma}
\begin{proof}
    We use induction on the length of the backbone.
    The base cases are backbones of length one, two or three.
    The spine of any realization is at most a polyline consisting of two segments and can therefore always be rotated to be $x$-monotone.
    The induction hypothesis is, that any lobster graph, with a backbone of length $k$ admits an $x$-monotone \wudc.
    In the induction step we need to show that any extension to a graph $G'$ with a backbone of length $k+1$, can be realized as a \wudc if and only if it can be realized in an $x$-monotone way.
    The extensions are done by appending a single new backbone vertex $v_{k+1}$, whose descendants are specified as a sorted list of the degrees of its direct neighbors.
    Since the total degree of every vertex is at most 6, the set $\Gamma$ of options for $v_{k+1}$ is constant (\cref{obs:dyn_prog_2}).
    Let $\Theta$ be the set of possible combinations of already occupied grid positions where $v_{k+1}$ is placed such that the spine remains $x$-monotone.
    Let $\Delta_6$ and $\Delta_3$ be the sets of possible placements of disks of descendants of $v_{k+1}$, when $v_{k+1}$ is placed at one of six (in the unrestricted case) or one of three (in the strictly $x$-monotone case) positions.
    Therefore we can enumerate all triples $(\gamma \in \Gamma, \delta_3 \in \Delta_3, \theta \in \Theta)$ and $(\gamma \in \Gamma, \delta_6 \in \Delta_6, \theta \in \Theta)$ and check if they are realizable.
    By exhaustive enumeration,\footnote{The cases were reduced, by considering symmetry and infeasibility beforehand. Enumeration was done in the form of a computer-assisted proof. Details are explained in \arxivTHENgd{\cref{sec:app:automated}}{Appendix D of the full version~\cite{bhore2021recognition}}} we have found that for every possible $(\gamma, \delta_6, \theta)$, which is realizable, we can find a suitable $(\gamma, \delta_3, \theta)$, which is realizable, too.
    This concludes the induction step.
\end{proof}

From \cref{lem:dyn_prog_monotone,lem:x_mon_lobster}, we conclude the following theorem.

\begin{theorem}\label{thm:dyn_prog_lobsters}
    It can be decided linear time if a lobster graph admits a \wudc on the triangular grid.
\end{theorem}

\section{Conclusions}
We have investigated the existing complexity gap for the recognition problem of \udrs and \wudcs.
In addition to the open problems for various graph classes in different settings (recall \cref{tab:results_state_open} in \cref{sec:intro}) there are two main open questions.
First, we have investigated \wudcs of lobsters on the triangular grid, however, it is not entirely clear if every lobster, which admits a \wudc, also does so on the grid.
Second, it seems reasonable to assume that our enumeration approach can be extended to graph classes beyond lobsters, which admit a \wudc at least on the triangular grid. In fact, we conjecture that every class of trees, in which each vertex has bounded distance to a central backbone in extension of caterpillars and lobsters, can be recognized in polynomial time by such an approach.
\subsubsection*{Acknowledgements.}
We thank 
Jonas Cleve and  Man-Kwun Chiu for fruitful discussions about the project during their research visits in Vienna.

\bibliography{39}
\clearpage
\arxivTHENgd{}{\end{document}}
\appendix

\section{Omitted Details of Section~\ref{hardness-section}}\label{hardness-appendix}

Bowen et al.~\cite{bdlrst-rscplrudct-15} proved that recognizing unit disk contact graphs is \textsf{NP}-hard for embedded trees, via a reduction from planar 3-SAT, which uses an auxiliary construction formulated as a realization of a polygonal linkage. 
Polygonal linkages are explained in \cref{sec:poly_linkage-appendix}.
The details of this auxiliary structure are explained in \cref{sec:auxiliary_structure-appendix}.
Then a tree, whose \udc is an approximation of the auxilliary structure and which mimics the shape and behaviour of it, is constructed.
This construction is summarized in \cref{sec:approx_construction-appendix}.

\subsection{Polygonal Linkages}\label{sec:poly_linkage-appendix}
Bowen et al.~\cite{bdlrst-rscplrudct-15} considered multiple problems in their work, one of which is the \emph{polygonal linkage realizability} (PLR) problem.
A polygonal linkage is a set $\mathcal{P}$ of convex polygons and a set $H$ of hinges. One hinge is a set of two or more points on the boundaries of distinct polygons.
A polygonal linkage is realizable, if every $p\in\mathcal{P}$ can be placed in the plane, s.t.
\begin{itemize}
	\item all polygons are interior disjoint
	\item for every hinge $h\in H$, all points of the hinge coincide and
	\item a predefined cyclic order of adjacent polygons around every hinge is kept.
\end{itemize}
In our case and in the case of this reduction, hinges are only of size two, i.e., a realization will identify exactly two points on distinct polygons per hinge.
This means that cyclic order around hinges is always kept by default.
A polygonal linkage and its realization are shown in \cref{fig:polyLinkageEx}.

\begin{figure}
	\centering
	\begin{minipage}[t]{.45\linewidth}
		\centering
		\includegraphics[page=1, width=.8\linewidth]{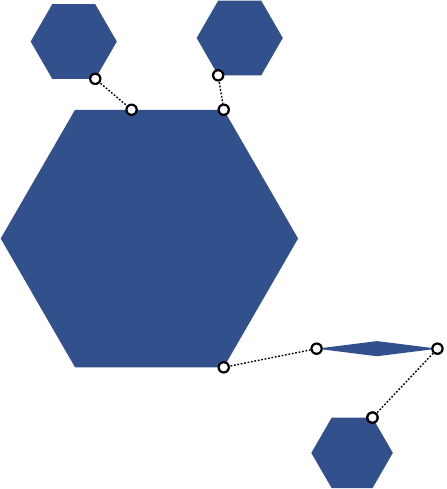}
		\subcaption{A set $\mathcal{P}$ of polygons with hinges}
		\label{fig:polyLinkageEx_a}
	\end{minipage}
	\quad
	\begin{minipage}[t]{.45\linewidth}
		\centering
		\includegraphics[page=2, width=.8\linewidth]{figures/polyLinkageExample.pdf}
		\subcaption{A realization of (a)}
		\label{fig:polyLinkageEx_b}
	\end{minipage}
	\caption{ 
		An instance of \textsc{PLR} (a), where points belonging to the same hinge are indicated with a dotted line connection. A placement of (rotated and/or translated) copies of the polygons in $\mathcal{P}$ in the plane is a realization (b) if the points of the same hinge are identified.
	}
	\label{fig:polyLinkageEx}
\end{figure}

\subsection{Auxiliary Structure}\label{sec:auxiliary_structure-appendix}
The auxilliary structure mimics a hexagonal grid.
This grid-like structure is obtained by using a hexagonal tiling of the plane and then shrinking every hexagon by a small amount to obtain narrow channels of a fixed height 
between two hexagons, where $\varepsilon$ is a sufficiently small constant.
At the corners of the hexagons, three such channels meet to form a junction.
The union of all channels and junctions forms the grid-like structure.
In this grid, a representation of the incidence graph $G_\phi$ of the planar 3-SAT instance $\phi$ is fitted, see \cref{fig:bowen_auxilliary_a}.

A variable $v$ of $\phi$ is represented in this grid as an alternating cycle of channels and junctions, indicated with a grey fill in \cref{fig:bowen_auxilliary_a}.
In such a cycle the channels can be filled with smaller hexagons of height $h$, which are connected to the large hexagon on the side of the channel -- which is on the ``inside'' of the variable cycle -- via a junction.
In a channel, one corner of each of the small hexagons is connected to the large hexagon via a hinge, s.t., the small hexagon can be ``flipped'' around this junction.
Due to the chosen size, the hexagon can be realized in one of two states, see \cref{fig:bowen_auxilliary_b,fig:bowen_auxilliary_d}.
The distance of the hinges of neighboring small hexagons is chosen in such a way that the state of one hexagon determines the state of all hexagons in the channel, see  \cref{fig:bowen_auxilliary_d}.
At each junction, we add an even smaller hexagon with a hinge to the corner of that large hexagon, which is adjacent to the channels on either side of the junction in the variable cycle.
This propagates the state of the hexagons in one channel through the junction into the other channel and so throughout the entire cycle.
See \cref{fig:bowen_auxilliary_b,fig:bowen_auxilliary_d} for a detailed explanation.

Wire gadgets are alternating paths of channels and junctions, which use the same mechanism to propagate the state of the hexagons in the channels and therefore admit two states overall.

Wire gadgets can be connected to a variable cycle at every junction using the third unoccupied channel and adding a second very small hexagon in the junction.
A wire gadget is considered to transmit the value \textit{true}, if and only if, part of the first small hexagon of the first channel of the wire gadget enters into the junction.
By placing the small hexagon on one or the other corner (cf. \cref{fig:bowen_auxilliary_b,fig:bowen_auxilliary_d}) the truth value which is transmitted can be ``inverted'' if necessary.

For every clause in $\phi$, three such wire gadgets are connected to the variable cycles of the occurring variables and the wires are routed to meet in a junction. This junction contains a small hexagon connected to the corner of a large hexagon via a long and very thin rhombus (in place of a line segment), s.t., an overlap-free realization is only possible, if at least one connected wire gadget has no hexagon entering the junction and is therefore in a true state, see \cref{fig:bowen_auxilliary_c}.

In order to guarantee a somewhat rigid placement of these hexagons, the entire construction is surrounded by a set of six huge hexagons, and the hexagons acting as the faces of the hexagonal grid are column-wise connected (cf. \cref{fig:bowen_hinged_aux_structure}). 
This restricts the position of the large hexagons to an $N$-neighborhood, where $N$ is a polynomial of the number of variables and clauses in $\phi$~\cite{bdlrst-rscplrudct-15}.

Note that the connections between the polygons in the polygonal linkage induce a tree.
In particular note that we can replace the hexagons with outerplanar graphs and replace the hinges between them with paths of length one to three and the entire construction remains an outerplanar graph.
The same way we can replace all hexagons with trees and replace the hinges with vertex paths of length one to three and the entire construction remains a tree.

\begin{figure}
	\centering
	\includegraphics[page=3, width=.6\linewidth]{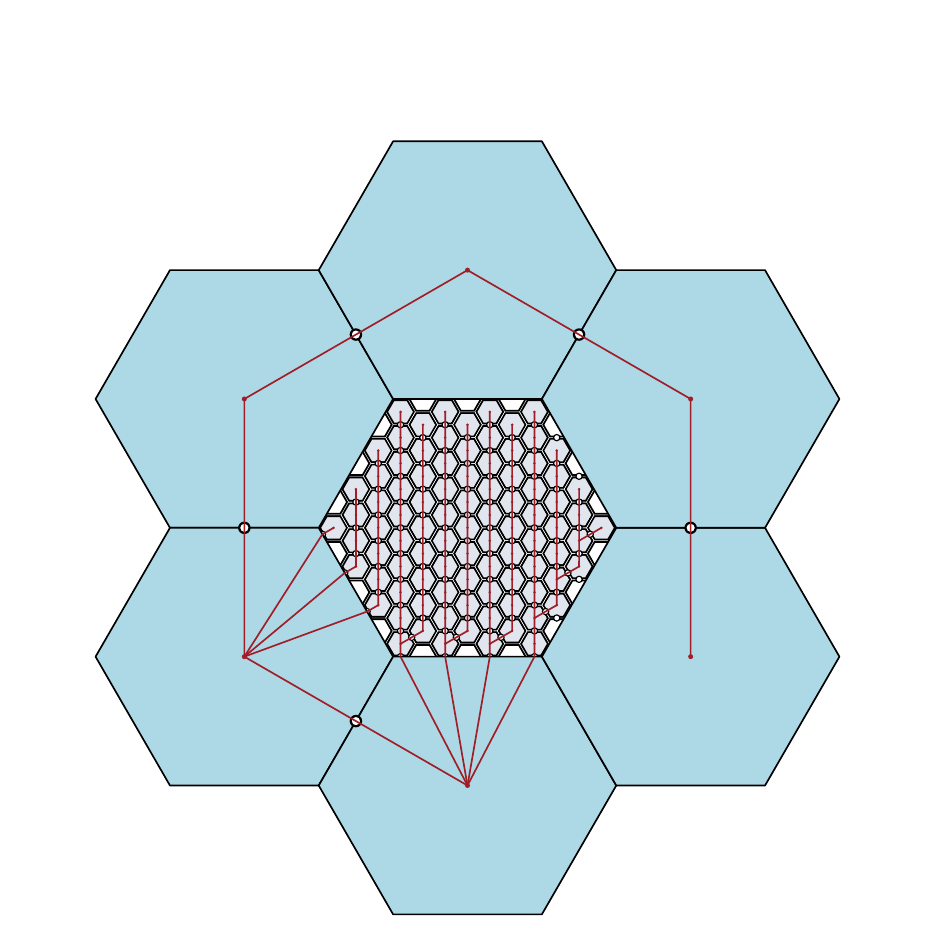}
	\caption{The composition of the rigid structure as a realization of a polygonal linkage.
		Hinges are indicated by points.
		The figure is a recreation of a similar figure in \cite{bdlrst-rscplrudct-15}, augmented with lines, which emphasize the tree-like structure of connections in the polygonal linkage.}
	\label{fig:bowen_hinged_aux_structure}
\end{figure}

For a more detailed description, a full construction and the proof of the semi-rigid placement we refer to the original paper of Bowen et al.~\cite{bdlrst-rscplrudct-15}.

\subsection{Approximating the Auxilliary Structure}\label{sec:approx_construction-appendix}
In order to prove the \NP-hardness for recognition of unit disk contact graphs, Bowen et al.~\cite{bdlrst-rscplrudct-15} created $\lambda$-stable approximations of the basic building blocks of the auxiliary structure (hexagons of varying sizes and long thin rhombi).
A graph is a $\lambda$-stable approximation of a polygon $P$, if the boundary union of all disks in its \udc is necessarily within a constant asymmetric Hausdorff distance $\lambda$ of a congruent copy of $P$.

Bowen et al.\ described the construction procedures for two graphs $T_k, T'_k$, which are $2$-stable approximations of a long thin rhombus and a regular hexagon. 
These two graphs are shown in \cref{fig:T_graphs}.
For the details of these constructions, we again refer to Bowen et al.~\cite{bdlrst-rscplrudct-15}.
The construction of our building blocks using outerplanar graphs in \cref{sec:outer} and embedded trees \cref{sec:trees} is based on these graphs.

\begin{figure}
	\centering
	\begin{minipage}[t]{.45\linewidth}
		\centering
		\includegraphics[page=1, width=\linewidth]{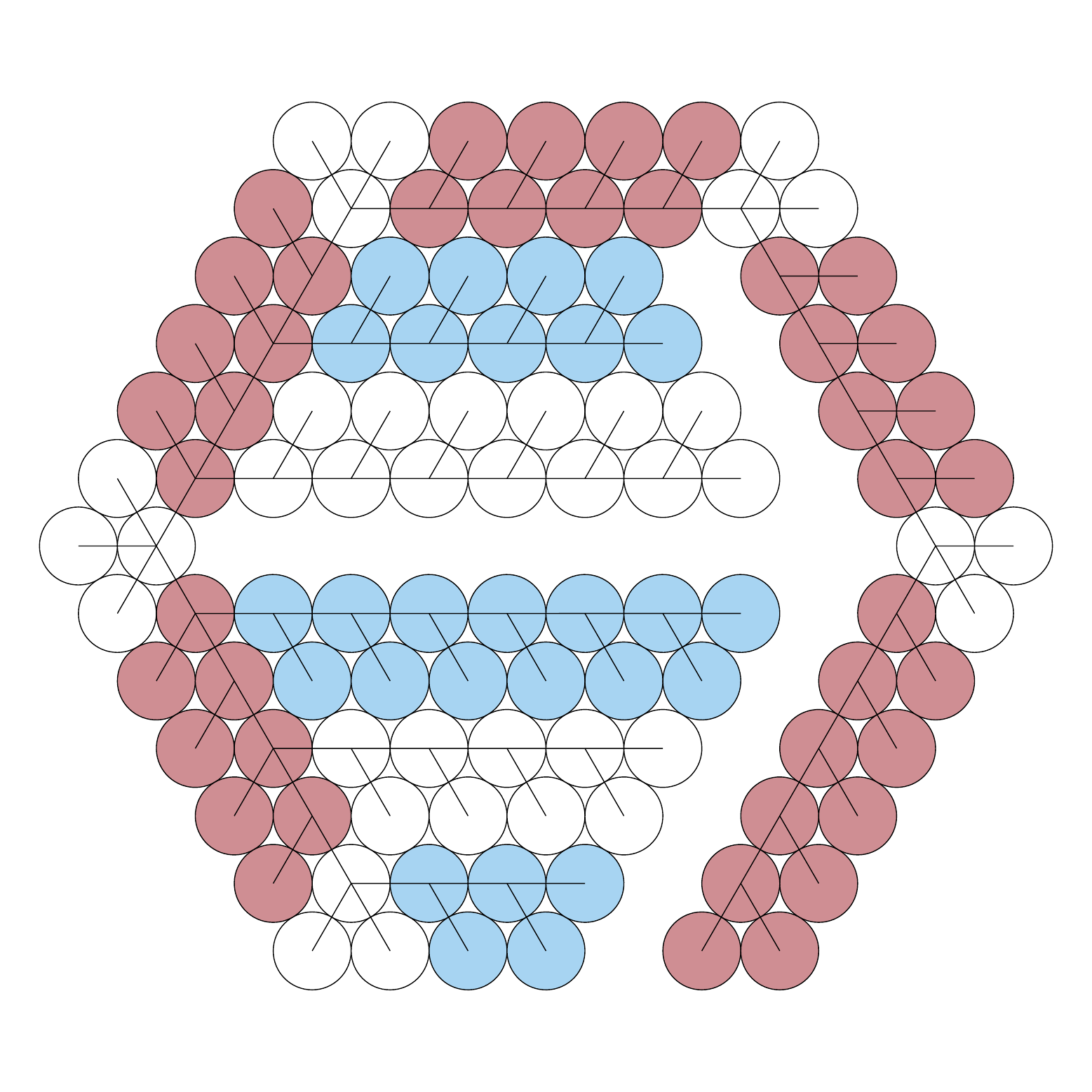}
		\subcaption{Graph $T_k$ and its \udc}
		\label{fig:T_k}
	\end{minipage}
	\quad
	\begin{minipage}[t]{.45\linewidth}
		\centering
		\includegraphics[page=2, width=\linewidth]{figures/lambdaStableBowen.pdf}
		\subcaption{Graph $T'_k$  and its \udc}
		\label{fig:T_prime_k}
	\end{minipage}
	\caption{
		The $2$-stable approximations of a long thin rhombus (a) and a regular hexagon (b).
		All Figures are recreations from Bowen et al.~\cite{bdlrst-rscplrudct-15}.
	}
	\label{fig:T_graphs}
\end{figure}

It remains to discuss how the hinges are modeled.
Two $\lambda$-stable approximations $G, G'$ of two polygons $P, P'$ are connected with a path of vertices of constant length, if there exists a hinge $h\in H$, with one point on the boundary of $P$ and the other on the boundary of $P'$.
The exact length of the path is dependent on the location of the hinge.
If the hinge is not placed on a corner of either polygon, they are simply connected via a single cut vertex, and with a path of length three otherwise, in order to facilitate more movement, which mimics the possibility of polygons to rotate around hinges.
The union of the \udc of two thus connected graphs $G, G'$, remains a constant factor approximation of a congruent copy of the realization of $P \cup P'$.

\subsection{Proofs of \cref{lem:rhombus_approx,lem:hexagon_approx}}\label{sec:app:outer_hardness}

To prove \cref{lem:rhombus_approx,lem:hexagon_approx} we need to establish some initial Lemmas.
First a ladder $L_k$ consists of two paths $v_1, v_2, \dots, v_k$ and $v_1', v_2', \dots, v_k'$ of vertices -- also called the \textit{outer} and \textit{inner} vertices respectively -- s.t., each pair $(v_i, v_i')$ is connected with an edge.
The blue vertices in \cref{fig:red_outer_graph_a} form a ladder, specifically $L_{17}$. A \udr of $L_{17}$ is shown in the blue disks of \cref{fig:red_outer_graph_c}.

\begin{lemma}
	\label{lem:ladder_embedding}
	Let $k\in \mathbb{N}$ and let $\mathcal{D}$ be a \udr of $L_k$.
	Then $\mathcal{E}_\mathcal{D}(L_k)$ induces the same embedding as shown in \cref{fig:red_outer_graph_a}, which is outerplanar.
\end{lemma}
\begin{proof}
	
	Clearly no disk can be placed inside a \udr of a $C_4$, without intersecting at least two disks of the $C_4$.
	Since every disk $d(v_i)$ of a vertex $v_i$, which is not part of the $C_4$ intersects at most one disk of the $C_4$, this is impossible and the outer face is fixed.
	Therefore the induced embedding is unique.
\end{proof}
Lemma~\ref{lem:ladder_embedding} implies that we can assign a clear ``outer side'', i.e., $d(v_1), \dots, d(v_k)$ and ``inner side'', i.e., $d(v'_1), \dots d(v'_k)$ to a \udr of $L_k$.

Next we want to ensure that any overall bend towards the outer side is impossible. 
For this we will augment $L_k$ with additional neighbors.
We restrict the ladders to an odd length and alternatingly add one and two leaf neighbors, which we will call extension neighbors.
The resulting graph can be seen in \cref{fig:red_outer_graph_a}.

\begin{lemma}\label{lem:bend_outer}
	Let $(v_1,v_2, v_3, v_4, v_5)$ be a chain of vertices on the outer side of a ladder $L_k$ with one, two, one, two and one extension neighbor respectively. Then the sum of angles $\gamma = \measuredangle v_3v_2v_1 + \measuredangle v_4v_3v_2 + \measuredangle v_5v_4v_3$ in a \udr of $L_k$ is smaller than $3\pi$, i.e., overall we bend more towards the inner side than towards the outer side.
\end{lemma}
\begin{proof}
	We will refer to the two extension neighbors of $v_2$ as $a$ and $a'$, both of which are placed on the outer side of the ladder.
	Similar we will call the extension neighbors of $v_4$, $b$ and $b'$.
	The singular extension neighbor of $v_3$, will be referred to as $c$.
	The naming is also shown in \cref{fig:bend_lemma_proof_positions}.
	Without loss of generality, we assume that the clockwise order on $v_2$ is $(v_1,a,a',v_3, v'_2)$ and analogue the clockwise order on $v_4$ is $(v_3, b, b', v_5, v_4')$.
	\cref{fig:bend_lemma_proof_positions,fig:bend_lemma_proof_realization} depict this placement and its \udr.
	
	\begin{figure}
		\centering
		\begin{subfigure}[t]{.24\linewidth}
			\centering
			\includegraphics{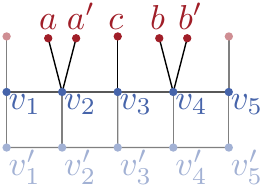}
			\subcaption{Arrangement}
			\label{fig:bend_lemma_proof_positions}
		\end{subfigure}
		\hfill
		\begin{subfigure}[t]{0.75\linewidth}
			\centering
			\includegraphics[page=6, scale=.9]{figures/bendLemmaProof.pdf}
			\subcaption{Increasing $\measuredangle v_4v_3v_2$ decreases $\measuredangle v_3v_2v_1$ and $\measuredangle v_5v_4v_3$}
			\label{fig:bend_lemma_proof_v3}
		\end{subfigure}
		\begin{subfigure}[t]{.24\linewidth}
			\centering
			\includegraphics[page=2]{figures/bendLemmaProof.pdf}
			\subcaption{\udr of (a)}
			\label{fig:bend_lemma_proof_realization}
		\end{subfigure}
		\hfill
		\begin{subfigure}[t]{.24\linewidth}
			\centering
			\includegraphics[page=3]{figures/bendLemmaProof.pdf}
			\subcaption{First position}
			\label{fig:bend_lemma_proof_extreme2}
		\end{subfigure}
		\hfill
		\begin{subfigure}[t]{.24\linewidth}
			\centering
			\includegraphics[page=4]{figures/bendLemmaProof.pdf}
			\subcaption{Second position}
			\label{fig:bend_lemma_proof_extreme1}
		\end{subfigure}
		\hfill
		\begin{subfigure}[t]{.24\linewidth}
			\centering
			\includegraphics[page=7]{figures/bendLemmaProof.pdf}
			\subcaption{Third position}
			\label{fig:bend_lemma_proof_extreme3}
		\end{subfigure}
		\caption{Figures for Lemma~\ref{lem:bend_outer}. 
			(a) Is the graph as stated in Lemma~\ref{lem:bend_outer}.
			(b) If $d(c)$ is moved outwards to facilitate an angle $\measuredangle v_4v_3v_2 = \pi + \varphi$, i.e., a local outwards bend, it enforces at the same time an inwards bend of $\varphi$ at $\measuredangle v_3v_2v_1$ and $\measuredangle v_5v_4v_3$.
			(c) A valid \udr of (a) including the $\varepsilon$ gaps between disks, when necessary.
			(d-f) Extreme positions of (c), s.t., $\measuredangle v_2v_3v_4 = \pi$ (d) or $\measuredangle v_2v_3v_4 = \frac{2\pi}{3}$ (e-f).
		}
	\end{figure}
	
	First observe that $\measuredangle v_1v_2a > \frac{\pi}{3}$, since $|{v_1v_2}|\leq 2, |{v_2a}|\leq 2$ and $|{v_1a}|> 2$.
	For the same reason we have $\measuredangle av_2a'> \frac{\pi}{3}$ and $\measuredangle a'v_2v_3> \frac{\pi}{3}$.
	This means that 
	
	\[
	\measuredangle v_1v_2v_3 = \measuredangle v_1v_2a + \measuredangle av_2a' + \measuredangle a'v_2v_3 > \pi
	\]
	
	and similarly $\measuredangle v_3v_4v_5 > \pi$.
	
	Now we will distinguish two extremes at $v_3$.
	In the first extreme position $d(c)$ is placed almost on top of $d(v_3)$, s.t., $\measuredangle v_2v_3c = \measuredangle cv_3v_4 = \alpha$ and $\alpha$ tends toward $\frac{\pi}{2}$ as $|v_3c|$ approaches 0.
	The extreme placement with $|v_3c|$ equal to 0 is shown in \cref{fig:bend_lemma_proof_extreme2}.
	This results in $\measuredangle v_2v_3v_4 > \pi$.
	To reach the second extreme (cf.~\cref{fig:bend_lemma_proof_extreme1}) $d(c)$ is moved outwards.
	As $|v_3c|$ increases, it facilitates a local outward bend and an angle $\measuredangle v_2v_3v_4 > \pi - \varphi$, where $\varphi = \cos^{-1}\left(\frac{|v_2c|^2 + |v_2v_3|^2 - |v_3c|^2}{2|v_2c||v_2v_3|}\right)$, see \cref{fig:bend_lemma_proof_v3}.
	However, in order to maximize this bend, $v_2$ and $v_4$ are moved into the resulting gap between $v_3$ and $c$, s.t., the distances $|v_2c|$ and $|v_4c|$ are equal to $2 + \varepsilon'$, where $\varepsilon'$ is a very small constant.
	In particular we have $\measuredangle a'v_2c > \frac{\pi}{3}$, and similarly $\measuredangle cv_4b> \frac{\pi}{3}$.
	Finally this leads to the following values:
	
	\begin{align*}
	\measuredangle v_1v_2v_3 &= \measuredangle v_1v_2a + \measuredangle av_2a' + \measuredangle a'v_2c + \measuredangle cv_2v_3 \geq \frac{\pi}{3} + \frac{\pi}{3} + \frac{\pi}{3} + \varphi = \pi + \varphi\\
	\measuredangle v_3v_4v_5 &= \measuredangle v_3v_4c + \measuredangle cv_4b + \measuredangle bv_4b' + \measuredangle b'v_4v_5 \geq \varphi + \frac{\pi}{3} + \frac{\pi}{3} + \frac{\pi}{3} = \pi + \varphi
	\end{align*}
	
	Clearly the overall bend towards the inside is minimized, if $\varphi$ and therefore $|v_3c|$ is minimal, i.e., we are as close as possible to the first extreme state.
	Since $|v_3c|>0$, we have $\measuredangle v_2v_3v_4 > \pi$ and finally
	
	\[
	\gamma = \measuredangle v_1v_2v_3 + \measuredangle v_2v_3v_4 + \measuredangle v_3v_4v_5 > 3\pi
	\]
	
	Note that the second extreme position can also be achieved by using more overlap between $d(v_2)$ or $d(v_4)$ and one of their extension neighbors, however, this configuration forces $\measuredangle av_2v_3$ and $\measuredangle v_3v_4b'$ to approach $\pi$ and therefore results in a similar overall inwards bend by $\varphi$, as shown in \cref{fig:bend_lemma_proof_extreme3}.
	
\end{proof}

We can also add a third extension neighbor to an outer vertex of a ladder to force an even more pronounced bend towards the inside.
The corresponding graph construction is shown in \cref{fig:red_outer_graph_b} and a valid \udr in \cref{fig:red_outer_graph_d}.
\begin{lemma}\label{lem:bend_outer_corner}
	Let $(v_1, v_2, v_3)$ be a chain of outside vertices of a ladder, s.t., $v_2$ has degree~5, i.e., 3 extension neighbors. Then the angle $\measuredangle v_3v_2v_1$ is smaller than $\frac{2\pi}{3}$, i.e., overall we bend at least $60^\circ$ towards the inner side.
\end{lemma}
\begin{proof}
	Note that the argument for the inward directed bend at $v_2$ and $v_4$ in the proof of Lemma~\ref{lem:bend_outer} was simply based on the number of extension neighbors.
	By increasing that number to 3, the lemma immediately follows.
\end{proof} 

We call this construction a corner connector.
The corner connector contains a 5-cycle, however similar to a 4-cycle, it is impossible to place a disk inside such a cycle without creating wrong adjacencies through overlap.

Finally we want to state some measure for the height of a \udr of an augmented $L_k$
\begin{lemma}\label{lem:ladder_height}
	The height the smallest bounding rectangle of a \udr of an augmented ladder $L_k$ is at least $2\sqrt{3}+2 - \varepsilon''$, where $\varepsilon''$ is a very small constant, which is the height of the smallest bounding box of an optimal packing of three rows of unit disks minus $\varepsilon''$. 
\end{lemma}
\begin{proof}
	The distance between the extension neighbors and the outer disks of the ladder has been argued in the proof of Lemma~\ref{lem:bend_outer}.
	Recall that the disks of the outside vertices of the ladder with a single extension neighbor are almost completely covered by it.
	Let $d(v_k)$ be such a covered disk.
	Its adjacent disk $d(v_k')$ on the inside of the ladder can therefore only overlap a very small amount with $d(v_k)$.
	Since $d(v'_k)$ is disjoint from $d(v_{k-1})$ and $d(v_{k+1})$, the extreme position (see \cref{fig:bend_lemma_proof_extreme2}) results in the smallest overall height of the ladder, which is the height of the dense packing of three rows of unit disks, i.e., $2\sqrt{3}+2$ minus the small amount $\varepsilon''$ which is attributed to the small possible overlap of $d(v_k')$ and $d(v_k)$.
\end{proof}

Now we have all necessities to describe the construction of the $7$-stable approximations of a long thin rhombus (Lemma~\ref{lem:rhombus_approx}) and a regular hexagon (Lemma~\ref{lem:hexagon_approx}).
Two ladders $L_k$ and $L_k'$ are placed opposite each other and connected on one end with three corner connectors as shown in \cref{fig:red_outer_rhombus}.
Note that the horizontal distance between $c_0$ and $v_1$ is four.
\rhombusouter*
\begin{proof}
	We can assume that $v_0 = (0,0)$ and that $c_0=(-2,0)$.
	We define a point $l_0 = (2k + 4, 0)$ on the positive $x$-axis with.
	
	We place a congruent copy of the rhombus, s.t., the long diagonal aligns with $\overline{c_0l_0}$, and one corner point coincides with the leftmost point of $c_0$, which is at $(-3,0)$.
	At every point the upper and lower boundary of the rhombus has a vertical distance of $\delta \in [0, 3\sqrt{3}+1]$ to $\overline{c_0l_0}$.
	Due to \cref{lem:bend_outer,lem:bend_outer_corner}, the overall bend of $L_k$ and $L'_k$ has to be towards the inner side.
	In fact the position depicted in the \cref{fig:red_outer_rhombus_a} is the limit and would require an additional infinitesimal bend inwards to be valid.
	A consequence is that $c_0$ and any $a_i$, have a maximal vertical distance of $3\sqrt{3}$ and therefore any point on the boundary of any disk has at most a vertical distance of $3\sqrt{3}+1 \approx 6.196 < 7$ to $\overline{c_0l_0}$.
	
	Clearly the \udrs of $L_k$ and $L'_k$ are completely overlap-free and therefore a tight packing as indicated in \cref{fig:red_outer_rhombus_b} is the narrowest possible configuration (including a small gap between the ladders).
	
	Now consider, that the smallest rectangular bounding box of the disks of the of outer vertices $v_1, \dots, v_k$ and $v_{-1}, \dots, v_{-k}$ (as labeled in \cref{fig:red_outer_rhombus_a}) and their extension neighbors have a combined height of at least $3\sqrt{3}+2$, which is the height of an optimal packing of four rows of unit disks.
	Therefore, at every point, there is part of the outline of $G^R_k$ on or above, and on or below ${c_0l_0}$ and the distance of the boundary of the union of all disks and the boundary of the rhombus is at most the distance of the rhombus itself to $\overline{c_0l_0}$, which as stated before, is $3\sqrt{3}+1 < 7$.
\end{proof}

Next we will create a $7$-stable approximation $G^H_k$ of a regular hexagon by tracing the outline of a hexagon with ladders and corner connectors as shown with the blue and the red disks in \cref{fig:red_complete_hex}.
Every side of the hexagon is represented by a $L_{k-3}$.
Let $L^1_j, L^2_j$ be two ladders, and let $v^k_i$ and $u^k_i$ be the $i$-th outer and inner vertices of $L^k_j$ respectively. 
At a corner, where $L^1_j$ and $L^2_j$ meet, we connect $u^1_j$ and $u^2_1$ with an edge.
Further, we connect the vertex $v_c$ of the corner connector with $v^1_j$ and $v^2_1$, thereby connecting the two ladders with a corner connector. The construction is depicted in \cref{fig:red_outer_graph_b}.
Six ladders are connected in this fashion through the use of five corner connectors, leaving the first and the last ladder disconnected.
This forms the outer hull of the hexagon, shown in \cref{fig:red_complete_hex}.

Since bends towards the outer side are impossible, we next need to control the possible deformation of a \udr of our construction towards the inner side.
To achieve this, we add ladders to our construction on the inside of the hull, reminiscent of our construction of $G^R_k$.
We chose the size of the hexagon, s.t., the inner area of a \udr of the hull admits the placement of $2m$ ladders ($m\in\mathbb{N}$ on the top and on the bottom, each) and that in a \udr, which is an optimal packing of disks, there is no space to add further disks in the gaps, without creating false adjacencies.

\hexagonapprox*
\begin{proof}
	By combining Lemmas~\ref{lem:bend_outer} and~\ref{lem:bend_outer_corner}, we conclude that a \udr of a such constructed outer hull does not exceed the boundary of a translated and/or rotated copy of the hexagon outlined as a dashed line in \cref{fig:red_complete_hex}, which has a side length of $2k-1$.
	
	Some movement towards the inside is possible.	
	By moving all ladders on the inside in the top half down by a complete row, including the ladder tracing the upper horizontal edge of the hexagon, we again arrive at an optimal packing of pairwise disjoint ladders, which contribute (due to \cref{lem:ladder_height}) at least their full height of $2\sqrt{3}+2 - \varepsilon''$ (again $\varepsilon''$ is a very small constant value) to the height of any \udr of the construction and theoretically, there is space to fold the right outer arm inward as depicted in \ref{fig:red_complete_hex_compact}.
	Even though the depicted folding of this outer arm into the optimal packing position is not possible, due to the construction, for increasingly long outer arms, we can get arbitrarily close to this packing.
	Therefore we analyze the (unreachable) worst case of the optimal packing, in which the farthest point from the boundary of the hexagon has a distance of $3\sqrt{3} + 1 + O(k)\cdot\varepsilon'' < 3\sqrt{3} + 1 < 7$, we conclude that no point on the boundary of the union of all disks in the \textsf{UDR} of a $G^H_k$ has a larger distance than 7 to a point on the boundary of the green dashed hexagon.
\end{proof}

\begin{figure}
	\centering
	\begin{subfigure}[t]{\linewidth}
		\centering
		\includegraphics[page=8, width=.75\linewidth]{figures/red_outer_hexagon.pdf}
		\subcaption{Compact state of $G^H_k$}
		\label{fig:red_complete_hex_compact}
		\vspace*{.2cm}
	\end{subfigure}
	\begin{subfigure}[t]{\linewidth}
		\centering
		\includegraphics[page=9, width=.75\linewidth]{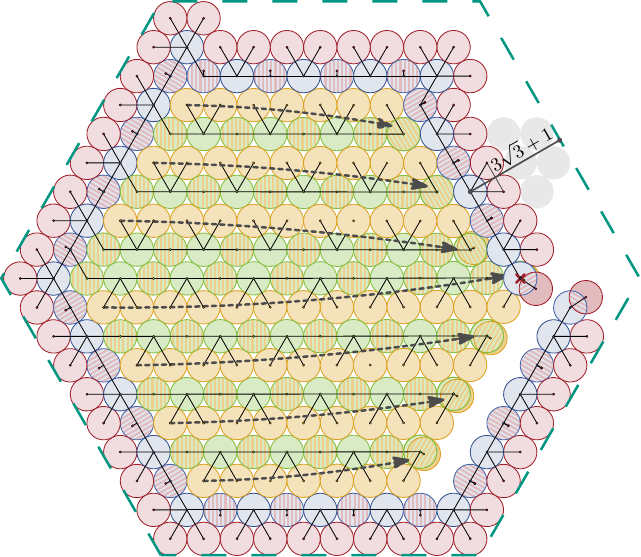}
		\subcaption{Compact state of $T^H_k$}
		\label{fig:red_complete_hex_compact_tree}
	\end{subfigure}
	\caption{In the most compact position, the \udr of $G^H_k$ and $T^H_k$ can -- similar to the rhombus -- fold all inner arms to one side (here to the lower side), and theoretically, there is space to fold the outer arms inward as depicted.
			The \udr of the arms are however again pairwise overlap free and the largest reachable distance any point on the boundary of the union of disks can have from the approximated regular hexagon is $3\sqrt{3} + 1$.
}
\end{figure}

\subsection{Proof of \cref{thm:tree_hardness}}\label{sec:app:tree_hardness}
\treehardness*
\begin{proof}
By modifying the construction of the previous section, we can prove that given a tree $T$ and an embedding $\mathcal{E}(T)$ it is \NP-hard to decide if $T$ admits a \textsf{UDR} $\mathcal{D}$, s.t., $\mathcal{E}_\mathcal{D}(T) = \mathcal{E}(T)$.
In particular, we construct two graphs $T^R_k$, $T^H_k$, which are both trees and $7$-stable approximations of a long thin rhombus and a regular hexagon respectively.

In the previous proof, we needed to construct ladders in a very specific way to force the placement of the disks of extension neighbors to one side.
This was done by adding the inner vertices, which made the graphs outerplanar.
Since in this setting the embedding is fixed, we can remove the inner vertices, which leaves the outer vertices and the extension neighbors.
Observe that the leftover graph is a tree.
The same holds for the corner connector.

Lemma~\ref{lem:bend_outer} and Lemma~\ref{lem:bend_outer_corner} can be directly applied to this new construction, since no assumption was made about the placement of the inner vertices outside of their function of forcing the disks of the extension neighbors to be placed on the outside.
This is now guaranteed by the embedding that is given as input.

To construct $T^R_k$, remove the inner vertices of a $G^R_k$.
Now the space between the two arms is large enough to accommodate one of the arms folding in on itself.
To prevent this, we add a third arm to the construction, as shown in \cref{fig:red_tree_rhombus_a}.
Recall that the outer arms of the construction cannot bend farther outwards than to the horizontal position in \cref{fig:red_tree_rhombus_a}.
Now again, no singular disk can be fit into the gap between the arms.

Next note that even without the inner vertices of the ladders, the three arms of a $T^R_k$ must together always be at least as high as an optimal packing of six rows of disks, which results in a height of $5\sqrt{3}+2$.
Therefore, at any given horizontal point, some part of the boundary of the union of all disks in a \udr of $T^R_k$ is either on or above $\overline{c_0l_0}$.
Since any point on the approximated rhombus is at most $3\sqrt{3}+1 < 7$ away from this line, $T^R_k$ is a $7$-stable approximation of a long thin rhombus.

We perform a similar adaption of $G^H_k$ to obtain $T^H_k$, a $7$-stable approximation of a regular Hexagon.
Analogously to $T^R_k$ the outer hull, which traces the hexagon is the same as for $G^H_k$, with all inner vertices removed from the ladders.
The relative placement of disks of extension neighbors to the outside is fixed, by the given embedding.
Therefore we can always place a congruent translated and/or rotated copy of the hexagon we want to approximate over the \udr of the construction shown in \cref{fig:red_complete_hex_tree}, s.t., the \udr is completely contained.

The exact placement of the inner ladder is slightly different as before.
We still fill the inside in such a manner, that the space between all upper and lower inner ladders does not admit the placement of an additional disk in a \udr, if realized as an optimal packing (cf. \cref{fig:red_complete_hex_tree}), however this requires an asymmetric construction, which places one additional ladder on one of the two sides, here to the bottom (similar to the construction of $T^R_k$).
Movement to the inside can be analyzed with a similar approach as before, by considering the closest placement to the inside of any of the arms, resulting in an optimal packing of the entire upper half of the construction, as seen in \cref{fig:red_complete_hex_compact_tree}.
Again this unobtainable worst case scenario leads to a maximal distance of $3\sqrt{3} + 1 < 7$ between the outline of any \udr of $T^H_k$ and the hexagon.

We therefore obtain again $5$-stable approximations of a long thin rhombus and a hexagon and hardness follows in the same fashion as for \cref{thm:outer_hardness}.
\end{proof}

\section{Omitted Details of \cref{algo-section:caterpillar}}
\subsection{Proof of Lemma~\ref{lem:claim_algo}\label{sec:proof_characterization}}

\claimalgo*
\begin{proof}
	
	Let $u_1, u_2, x$ (resp. $v_1,v_2, x'$) be the first, second and third neighbor of $u$ (resp. $v$) in counterclockwise/clockwise order, respectively. 
	We assume a placement of $d(u_1), d(u_2), d(v_1)$ and $d(v_2)$ as depicted in \cref{fig:claim_algo}.
	Since $x$ is adjacent to $u$ and not adjacent to $u_2$ and $v$, the possible space of placing the center of $d(x)$ is restricted to a small open fan-shaped area (see \cref{fig:claim_algo_c}).
	A symmetric area exists for $d(x')$ as well.
	Moreover, all three pairs of points with a distance of at least two, namely \cref{fig:claim_algo_d,fig:claim_algo_f} and a symmetric version of \cref{fig:claim_algo_f}, are on the boundaries of these areas.
	Therefore no such pair exists in the open fan-shaped areas.
	Note that a similar argument holds for the lower side and any rotation of $d(u_1), d(u_2), d(v_1)$ or $d(v_2)$ will decrease the area on one side. This completes the proof.
\end{proof}

\begin{figure}
	\centering
	\begin{subfigure}[t]{.32\linewidth}
		\centering
		\includegraphics[page=2, width=.95\linewidth]{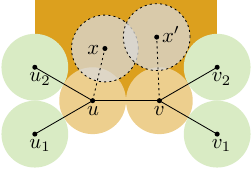}
		\caption{Two disks to be placed}
		\label{fig:claim_algo_a}
	\end{subfigure}
	\begin{subfigure}[t]{.32\linewidth}
		\centering
		\includegraphics[page=3, width=.95\linewidth]{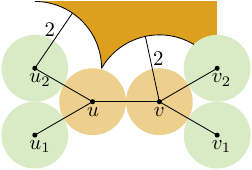}
		\subcaption{$|xu_2|, |xv| > 2$}
		\label{fig:claim_algo_b}
	\end{subfigure}
	\begin{subfigure}[t]{.32\linewidth}
		\centering
		\includegraphics[page=4, width=.95\linewidth]{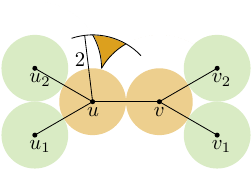}
		\subcaption{$|xu| \leq 2$}
		\label{fig:claim_algo_c}
	\end{subfigure}
	\begin{subfigure}[t]{.32\linewidth}
		\centering
		\includegraphics[page=5, width=.95\linewidth]{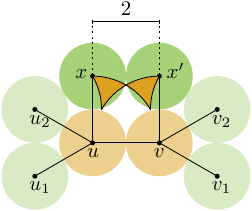}
		\caption{First extreme}
		\label{fig:claim_algo_d}
	\end{subfigure}
	\begin{subfigure}[t]{.32\linewidth}
		\centering
		\includegraphics[page=6, width=.95\linewidth]{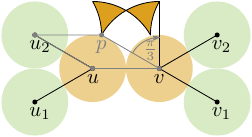}
		\subcaption{$\Parallelogramm uvpu_2$ is a rhombus}{-1}
		\label{fig:claim_algo_e}
	\end{subfigure}
	\begin{subfigure}[t]{.32\linewidth}
		\centering
		\includegraphics[page=7, width=.95\linewidth]{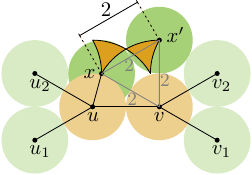}
		\subcaption{Second extreme}
		\label{fig:claim_algo_f}
	\end{subfigure}
	\caption{Visualization of the proof steps of Lemma~\ref{lem:claim_algo}.
		The area highlighted in orange illustrates possible placements for the center of $d(x)$, which is progressively restricted by its adjacent and non-adjacent neighbors.}
	\label{fig:claim_algo}
\end{figure}

\subsection{Correctness}\label{sec:app:algo_correct}
If a caterpillar contains consecutive degree 5 vertices we reject it as it has no \udr.
In any other case, the algorithm above can represent it in a way, s.t., at any point the angle between a backbone vertex and its left most-top and left-most bottom leaf is less than $\pi$. As long as this property holds, we can always add a new backbone vertex.
Moreover, if the sequence is extended by a backbone vertex of at most degree 4, this property immediately holds again.
If the sequence is extended by a backbone vertex of degree 5, a backbone vertex of degree at most 4 must follow. 
Once we place this degree~4 vertex appropriately, cf. \cref{fig:caterpillar_construction_b}, then the property immediately holds again.
\section{Omitted Details of Section~\ref{algo-section:lobster}}

\dynprogB*
\begin{proof}
	The vertex $v$ can only be placed at one of three possible grid positions.
	For each position $p$ at which $v$ can be placed, every placement of $v$ and its descendants corresponds to a subset of the constant size set $S(p)$, which yields again a constant number of placements.
\end{proof}

\dynprogD*
\begin{proof}
    Clearly the number of grid positions inside a circle of radius 2 are constant.
    Therefore only a constant number of previously placed backbone vertices could have possibly been placed a this grid position or placed one of their descendants on it.
    Since the \wudc is strictly $x$--monotone, every backbone vertex must be placed to the right of the previous one.
    In particular that means that only a constant number of backbone vertices can have a distance smaller or equal to 2 in $x$ direction to the grid position and moreover, the graph distance between the first and last such vertex has to be constant.
    Since every graph, which possible can occupy the grid position has to have an $x$-distance of less or equal to 2, this concludes the proof.
\end{proof}

\dynprogE*
\begin{proof}
    Let $U$ be the \wudc of $A$ and $C$.
    Now we take the subpart of $U$ induced by vertices of $C$.
    Assume w.l.o.g., that the disk of the first backbone vertex of $C$ was placed one grid position to the right of the disk of the last backbone vertex of $B$.
    We now translate this part, s.t., the disk of the first backbone vertex of $C$ is placed one grid position to the right of the last backbone vertex of $B$.
    Since $B$ and $A$ had the exact same set of already occupied grid positions, this is a valid \wudc of $B$.
\end{proof}

\dynprogmonotone*
\begin{proof}
    The dynamic program starts at the first backbone vertex $V_1$, which is w.l.o.g., placed at a fixed position and enumerates all possible placements of the descendants of $v_1$.
    Due to \cref{obs:dyn_prog_2}, this can be done in $O(1)$ time.
    Every single placement yields a pattern of occupied positions on the grid.
    Note that two placements, which yield the same occupation pattern and place $v_1$ at the same position relative to that pattern (up to translation and rotation) are indistinguishable, when considering the obstruction they pose for the rest of the \wudc, due to \cref{obs:dyn_prog_5}.
    We save all possible occupation patterns together with the position of the previously placed backbone vertex as a record.
    After processing the first backbone vertex, we therefore clearly have a constant number of records.
    
    When the dynamic program continues on from the $(i-1)$-th to the $i$-th backbone vertex $v_i$ it needs to place the new backbone vertex at a specific position.
    Every such position $p$ has a constant set of grid positions, which are at most at a grid distance of 2 to $p$.
    Note that these grid positions can be occupied by descendants of previous backbone vertices, however, due to \cref{obs:dyn_prog_3} these backbone vertices have at most a constant distance in the graph to $v_i$.
    Therefore the dynamic program only needs to remember the exact occupation pattern for a constant number of previously placed backbone vertices.
    
    In particular this means, that for all possible occupation patterns (a constant number per backbone vertex) of all relevant previous backbone vertices (a constant number) we need to save all possible placements (a constant number per backbone vertex and occupation pattern).
    We call one such saved combination, together with the position of the previously placed backbone vertex, a \emph{record} and such a record is clearly of constant size.
    
    For every single record, we need to enumerate all possible placements of the descendants of $v_i$.
    Due to \cref{obs:dyn_prog_2}, this is also a constant number.
    Finally we can check if the placement is valid for this record in constant time.
    For every record, we now save for every possible valid placement of $v_i$ and its descendants only the information of the occupied grid positions for $v_i$ and the required constant number of predecessors of the backbone.
    The crucial observation is that we do not need to differentiate between two records, which are equivalent as defined in \cref{obs:dyn_prog_5}.
    If we do consider records, whose last $C$ backbone vertices induce the exact same pattern and whose last backbone vertex was placed at the exact same position (up to translation), then it is clear that only a constant number of unique records exist for a single backbone vertex.
    In particular, $v_{i+1}$ starts only with a constant number of records to check.
    This results in only an overall constant time to process a single backbone vertex.
\end{proof}

\section{Computer Assisted Proof of \cref{lem:x_mon_lobster}}\label{sec:app:automated}
Here we aim to provide a detailed overview of our computer assisted approach of the induction step for the proof of \cref{lem:x_mon_lobster}.
Recall that the induction hypothesis is as follows.
Any lobster graph $G$, with a backbone of length $k$ admits a strictly $x$-monotone \wudc.
In the induction step we need to prove that any extension to a lobster graph $G'$ with a backbone of length $k+1$ also admits an $x$-monotone \wudc.
Let $v_1, \dots, v_{k+1}$ be the vertices of $B_{G'}$.
Note that $v_{k+1}$ needs to be placed, s.t., it is in contact with $v_k$ and, in a strictly $x$-monotone \wudc, $v_{k+1}$ can be placed at one of three possible grid positions $p_1, p_2, p_3$, while in a \wudc it can be placed at one of six possible grid positions $p_1, \dots, p_6$ (at least one of which will already be occupied by $v_k$).
If we can find an instance of a graph $G'$, s.t., $v_{k+1}$ is forced to be placed at one of the positions $p_4, p_5, p_6$ in order to realize a \wudc of $G'$, we know that the Theorem would be false.
The crucial observation is that we can enumerate for a given graph $G$ all possible extensions $G'$, since the maximum degree of every vertex is naturally bounded by six.
Moreover the exact structure of $G$ is not relevant, when investigating if $G'$ admits a \wudc.
But rather, we want to know to what extent the already placed disks of $G$ can interfere with a placement of disks belonging to $v_{k+1}$ and its neighbors.

\subsection{Enumerating All Combinations of Free Grid Positions}
Note that the set of grid positions, which can be used to place $v_{k+1}$ or one of its neighbors is a constant size set of 29 positions, see \cref{fig:app:grid_pos_set}.

\begin{figure}
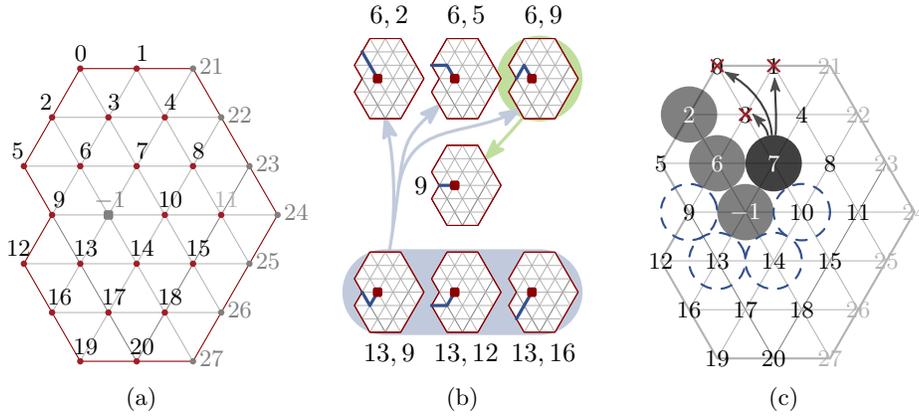

    \centering
    \begin{subfigure}[t]{.3\linewidth}
        \centering
        \includegraphics[page=1]{figures/Lobsters/automated_appendix_details.pdf}
        \subcaption{}
        \label{fig:app:grid_pos_set}
    \end{subfigure}
    \hfill
    \begin{subfigure}[t]{.3\linewidth}
        \centering
        \includegraphics[page=2]{figures/Lobsters/automated_appendix_details.pdf}
        \subcaption{}
        \label{fig:app:backbone_paths}
    \end{subfigure}
    \hfill
    \begin{subfigure}[t]{.3\linewidth}
        \centering
        \includegraphics[page=12]{figures/Lobsters/automated_appendix_details.pdf}
        \subcaption{}
        \label{fig:app:single_impl}
    \end{subfigure}
    \caption{
    	(a) A constant size set of grid positions, which could be occupied by $v_{k+1}$ or one of its neighbors for a fixed position of $v_k$ (at $-1$) and a possible placement of $v_{k+1}$ at either 7, 10 or 14.
    	Positions, which are necessarily free (assuming an already $x$-monotone placement of $v_1$ up to $v_k$) are $21, \dots, 27$.
    	(b) All cases, how the previous backbone could be connected to the position of $v_k$.
    	The numbers indicate the occupied positions of the previous backbone vertices in each case.
    	The top right case is covered, because the center case is a subset of it.
    	The bottom three cases are covered by symmetry.
	}
    \label{fig:app:proof_details}
\end{figure}

We assume w.l.o.g. that $v_k$ was placed at grid position $-1$.
The positions we need to consider, when looking at obstructions is a smaller set than the full 29, as positions $21,\dots,27$ can not possibly be occupied, since they have a minimal grid distance of three to any possible previous backbone vertex position (Note that we denote the grid distance between two neighboring grid positions as one, even though the actual distance has to be two to accommodate the unit disks).
Further we can assume $-1$ to be blocked as $v_k$ is placed at this position

Next we can assume that the grid position $-1$ has to be connected to the left hand side of the area outlined in red in see \cref{fig:app:grid_pos_set}, as there must be at the very least an incoming path of backbone vertices.
Every such $x$-monotone path on the grid starting at any point to the left of $-$ on the boundary of the red region and ending at $-1$ has a maximum length of three.
Note that our base cases cover all backbone lengths up to a length of three.
Various possibilities exist for such a path, however we only need to consider three, see \cref{fig:app:backbone_paths}.
The top right case is covered, because it is indistinguishable from the center case, when a single neighbor of either $v_k$ or $v_{k-1}$ is placed at position 6.
The bottom three cases are simply symmetric to the top three cases.

This now leaves us with two cases with a set of 19 possibly occupied grid positions and one with 20.
From here on out we assume that we represent these grid positions as three lists of boolean variables $L_C, C\in \{\{6,2\}, \{6,5\}, \{9\}\}$ of length 21, where $L_c[i]$ ($\neg L_c[i]$) means that the position is (not) occupied.
Clearly some of these positions are fixed to be true, e.g., $c\in C \implies L_C[c]$, i.e., the grid positions, which are occupied by disks of the previous backbone vertices are guaranteed to be blocked.
Moreover, some of the positions are also guaranteed to be not occupied.
For example the grid positions 19 and 20 do not have any path to either $-1$, 2, 6, or any possible position of a previous backbone disk - recall that the realization is $x$-monotone up to $v_k$ - of grid distance 2 or shorter and can therefore not be occupied by a previously placed disk.
All possible combinations of the values in these lists, still leaves us with two times $2^{19}$ plus $2^{20}$ (since $L_{\{6, 2\}}$ and $L_{\{6, 5\}}$ have two blocked spots and $L_{\{9\}}$ has one) cases.

We can reduce the number even further when considering the implication of one true value on all other values in a list $L_C$.
For this we first set every singular entry in a list to true and consider its implications.
One example is shown in \cref{fig:app:single_impl}, where we consider the implications of setting $L_{\{6, 2\}}[7]$.
Since $7$ is occupied, the next backbone disk can only be placed at positions 9, 10, 13 or 14. Note that positions 0, 1 and 3 do not have a free path (that is a sequence of free grid position), s.t., the start and end point of this path have grid distance two to either one of these positions.
Therefore, They are unusable to any descendent of $v_{k+1}$ in a \wudc.
We can consider them occupied.
This implication can be formulated as 

\[ L_{\{6, 2\}}[7] \implies L_{\{6, 2\}}[0] \land L_{\{6, 2\}}[1] \land L_{\{6, 2\}}[3]\]

We can manually consider every single position in the grid and their implications on all other positions in each of the three scenarios.
All such implications are visualized as an implication matrix in \cref{fig:app:Implications_single_1,fig:app:Implications_single_2,fig:app:Implications_single_3}.

\begin{figure}
	\centering
	\begin{subfigure}{\linewidth}
	    \centering 
	    \includegraphics[page=6, width=.87\linewidth]{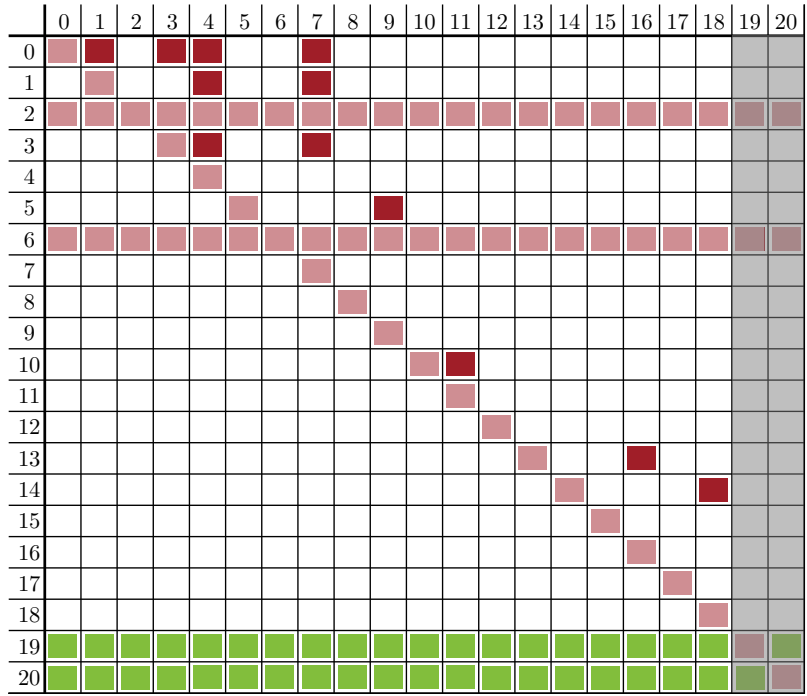}
	    \subcaption{Implications of singular grid positions}
    	\label{fig:app:Implications_single_1}
	\end{subfigure}
	\begin{subfigure}{\linewidth}
	    \centering 
	    \includegraphics[page=9, width=\linewidth]{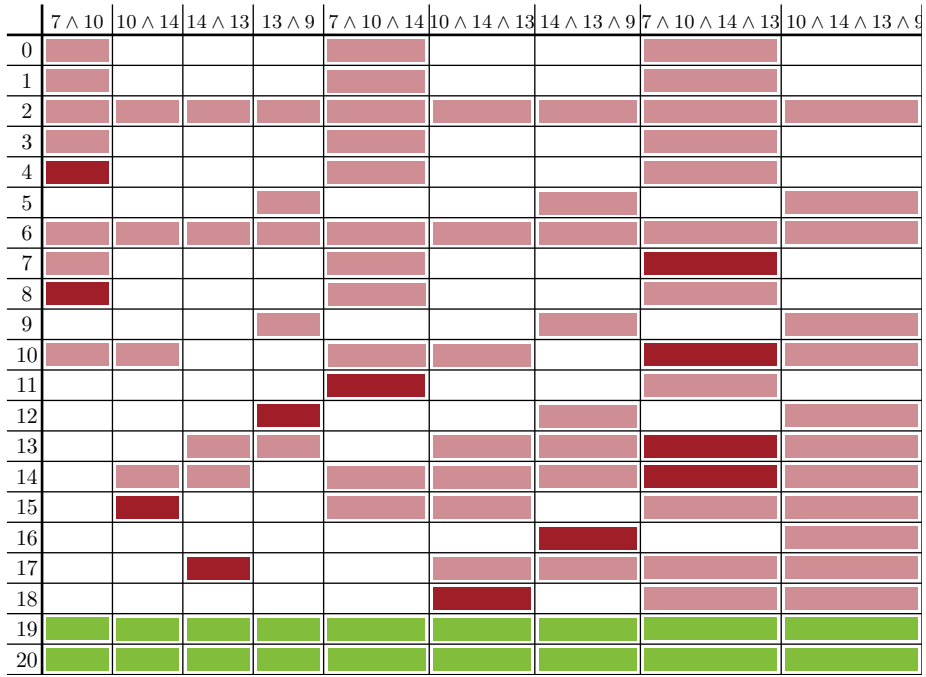}
	    \subcaption{Implications of sets of grid positions}
    	\label{fig:app:Implications_set_1}
	\end{subfigure}
	\caption{Implications of singular grid positions being occupied on all other grid positions in $L_{\{6, 2\}}$.
	Red blocks indicate, the column number implies the row number.
	Light red blocks indicate a trivial or previously determined implication.
	Green blocks indicate guaranteed free positions.
	Grey overlay indicates that the implication of these columns was not considered (as they are free).
	}
\end{figure}

\begin{figure}
	\centering
	\begin{subfigure}{\linewidth}
	    \centering 
	    \includegraphics[page=7, width=.87\linewidth]{figures/Lobsters/automated_appendix_details.pdf}
	    \subcaption{Implications of singular grid positions}
	    \label{fig:app:Implications_single_2}
	\end{subfigure}
	\begin{subfigure}{\linewidth}
	    \centering 
	    \includegraphics[page=10, width=\linewidth]{figures/Lobsters/automated_appendix_details.pdf}
	    \subcaption{Implications of sets of grid positions}
    	\label{fig:app:Implications_set_2}
	\end{subfigure}
	\caption{Implications of singular grid positions being occupied on all other grid positions in $L_{\{6, 5\}}$.
	Red blocks indicate, the column number implies the row number.
	Light red blocks indicate a trivial or previously determined implication.
	Green blocks indicate guaranteed free positions.
	Grey overlay indicates that the implication of these columns was not considered (as they are free).}
\end{figure}

\begin{figure}
	\centering
	\begin{subfigure}{\linewidth}
	    \centering 
	    \includegraphics[page=8, width=.87\linewidth]{figures/Lobsters/automated_appendix_details.pdf}
	    \subcaption{Implications of singular grid positions}
	    \label{fig:app:Implications_single_3}
	\end{subfigure}
	\begin{subfigure}{\linewidth}
	    \centering 
	    \includegraphics[page=11, width=\linewidth]{figures/Lobsters/automated_appendix_details.pdf}
	    \subcaption{Implications of sets of grid positions}
    	\label{fig:app:Implications_set_3}
	\end{subfigure}
	\caption{Implications of singular grid positions being occupied on all other grid positions in $L_{\{9\}}$.
	Red blocks indicate, the column number implies the row number.
	Light red blocks indicate a trivial or previously determined implication.
	Green blocks indicate guaranteed free positions.
	Grey overlay indicates that the implication of these columns was not considered (as they are free).}
\end{figure}

We further considered the implications of some sets of grid positions, e.g., $L_{\{6, 2\}}[7] \land L_{\{6, 2\}}[10] \implies L_{\{6, 2\}}[4]$, i.e., in the $L_{\{6, 2\}}$ setting, if both 7 and 10 are occupied, 4 is unreachable and can therefore be considered occupied.
All such implications are again visualized in an implication matrix in \cref{fig:app:Implications_set_1,fig:app:Implications_set_2,fig:app:Implications_set_3}.

We then generate all possible combinations of Boolean variables for the three lists $L_{\{6, 2\}}$, $L_{\{6, 5\}}$ and $L_{\{9\}}$, but omit any combination, which contradicts any of the implications, we have identified, e.g., if any combination has 
\[L_{\{6, 2\}}[7] \land \neg L_{\{6, 2\}}[0]\]
then we can omit it and similar for all other implications.

\subsection{Enumerating all possible Extensions of $G$ to $G'$}
The next crucial observation is that the number of the neighbors of $v_{k+1}$ and the number of their neighbors is bounded by 5, i.e., $v_k{k+1}$ can have at most 5 neighbors (excluding $v_k$), which can have at most 5 neighbors themselves (excluding $v_{k+1}$).
Again we represent $v_{k+1}$ as a sorted list $L_v$ of length equal to the number of neighbors (excluding $v_k$), s.t., for every neighbor $v_n$ of $v_{k+1}$, there is an $i$, s.t., $L_v[i]$ is equal to the number of neighbors of $v_n$ (excluding $v_{k+1}$).
For example, the first backbone vertex of \cref{fig:intro_c}, would be represented as $[3, 2, 1]$.

Now it is clear that we can employ again an exhaustive enumeration to construct every sorted list of length 5 or lower, with all possible combination of values between $0$ and $5$.
We can however again identify some exclusion criteria, based on structures, which make a graph unrealizable outright.
These criteria are:
\begin{itemize}
    \item If the total number of vertices is larger than the not occupied entries in $L_c$, the graph is not realizable
    \item Depending on the number of neighbors (excluding $v_{k+1}$) we can identify, how many grid positions at grid distance 1 to the placement of $v_{k+1}$, a single neighbor will occupy
    \begin{itemize}
        \item If $L_v[i]$ equals 5, it will occupy at least 3 such spaces, (2 additionally)
        \item If $L_v[i]$ equals 4, it will occupy at least 2 such spaces (1 additionally)
        \item If there exists $i, i', i'', i'''$, s.t., $L_v[i], L_v[i'], L_v[i'']$ are equal or greater than 3 and $L_v[i''']$ is equal or greater than 2, these neighbors will together occupy at least 5 such spaces (1 additionally)
    \end{itemize}
\end{itemize}
\begin{figure}
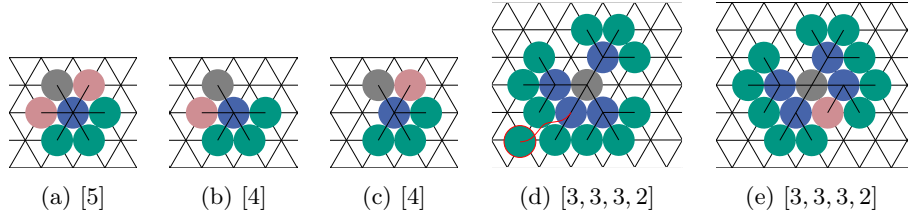

    \centering
    \begin{subfigure}[t]{.16\linewidth}
        \centering
	    \includegraphics[page=13, scale=0.75]{figures/Lobsters/automated_appendix_details.pdf}
	    \subcaption{$[5]$}
	    \label{fig:app:additional_slots_a}
    \end{subfigure}
    \hfill
    \begin{subfigure}[t]{.16\linewidth}
        \centering
	    \includegraphics[page=14, scale=0.75]{figures/Lobsters/automated_appendix_details.pdf}
	    \subcaption{$[4]$}
	    \label{fig:app:additional_slots_b}
    \end{subfigure}
    \hfill
    \begin{subfigure}[t]{.16\linewidth}
        \centering
	    \includegraphics[page=15, scale=0.75]{figures/Lobsters/automated_appendix_details.pdf}
	    \subcaption{$[4]$}
	    \label{fig:app:additional_slots_c}
    \end{subfigure}
    \hfill
    \begin{subfigure}[t]{.23\linewidth}
        \centering
	    \includegraphics[page=16, scale=0.75]{figures/Lobsters/automated_appendix_details.pdf}
	    \subcaption{$[3, 3, 3, 2]$}
	    \label{fig:app:additional_slots_d}
    \end{subfigure}
    \hfill
    \begin{subfigure}[t]{.23\linewidth}
        \centering
	    \includegraphics[page=17, scale=0.75]{figures/Lobsters/automated_appendix_details.pdf}
	    \subcaption{$[3, 3, 3, 2]$}
	    \label{fig:app:additional_slots_e}
    \end{subfigure}
    \caption{Additionally occupied grid positions at grid distance 1 from the position of $v_{k+1}$ (grey).
    A descendant, which is non-adjacent to $v_{k+1}$, but occupies a grid position at grid distance 1 to $v_{k+1}$ is colored in light red.
    The realization in (d) is not valid.
    The subcaptions are according to the list notation of the neighbors as used in the text.}
    \label{fig:additional_slots}
\end{figure}

If any of the first two criteria applies or the number of neighbors of $v_{k+1}$ plus the number of additionally occupied grid positions (indicated in parenthesis above) is greater than the number of free spaces at grid distance 1 to the chosen position of $v_{k+1}$, we can disregards this sorted list.

With these tools, we can now generate the set of instances $I = (\Gamma\times\Theta)\setminus \Xi$, where $\Xi$ is any combination of an element $\gamma \in \Gamma$ and $\theta \in \Theta$, s.t., either $\gamma$ or $\theta$ pr the combination of both could be excluded based on the criteria outlined above.
The size of $I$ is $|I| = 7017913$.

\subsection{Enumerating All Possible Realizations in a Given Setting}
Finally we use an exhaustive enumeration of possible placement of the descendants of $v_{k+1}$, for each of the three or six possible placements of $v_{k+1}$ to generate $\Delta_3$ and $\Delta_6$.
For this we consider a fixed position of $v_{k+1}$.
Then we consider every possible placement of its neighbors on the grid positions at grid distance 1 to $v_{k+1}$.
If no such placement exists, the setting is not realizable.
Next for a neighbor $v_n$ of $v_{k+1}$ we consider all placements of its neighbors (excluding $v_{k+1}$) at the still free grid positions.
If this does not exist for any $v_n$, the setting is not realizable.
Now we have multiple sets of possible realizations of the neighbors of every $v_n$.
We finally consider every combination of choosing one element out of each of these sets, s.t., no two disks occupy the same grid positions.
If no such combination exists, the setting is not realizable.
Otherwise the set of all combinations is a complete enumeration of all realizations of $v_{k+1}$ and all of its descendants.

\end{document}